\documentclass[acmsmall,nonacm,screen]{acmart}

\setcopyright{none}
\acmJournal{TEAC}
\acmYear{2026}
\acmVolume{0}
\acmNumber{0}
\acmArticle{0}
\acmDOI{}

\usepackage{mathtools}
\usepackage{booktabs}

\AtEndPreamble{%
  \theoremstyle{acmplain}%
  \newtheorem{assumption}[theorem]{Assumption}%
  \theoremstyle{acmdefinition}%
  \newtheorem{remark}[theorem]{Remark}%
}

\newcommand{\R}{\mathbb{R}}
\newcommand{\E}{\mathbb{E}}

\newcommand{\calG}{\mathcal{G}}
\newcommand{\calR}{\mathcal{R}}
\newcommand{\calI}{\mathcal{I}}

\newcommand{\calY}{\mathcal{Y}}
\newcommand{\calN}{\mathcal{N}}
\newcommand{\Cbind}{\mathcal{C}_{\mathrm{bind}}}
\newcommand{\Sbar}{\bar{S}}

\numberwithin{equation}{section}

\newcommand{\suppsec}[1]{the full-proof block at the end of this appendix}
\newcommand{\suppnote}{}
\newcommand{\multidimnote}{Appendix~\ref{sec:multi-dim-geom} develops these statements and the obstructions behind them.}

\title[The Endogeneity of Miscalibration]{The Endogeneity of Miscalibration: \\
  Impossibility and Escape in Scored Reporting}

\author{Lauri Lov\'{e}n}
\orcid{0000-0001-9475-4839}
\authornote{Corresponding author.}
\affiliation{%
  \institution{Future Computing Group, University of Oulu}
  \city{Oulu}
  \country{Finland}
}
\email{lauri.loven@oulu.fi}

\author{Sasu Tarkoma}
\affiliation{%
  \institution{University of Oulu and University of Helsinki}
  \city{Oulu / Helsinki}
  \country{Finland}
}

\renewcommand{\shortauthors}{Lov\'{e}n and Tarkoma}



\begin{document}

\begin{abstract}
An agent's probability report is paid for twice: by a strictly proper scoring rule, which rewards accuracy, and by an approval rule for the decision it triggers. An approval rule that is not affine is known to defeat truthful reporting. We show the conflict is built in: whenever perfect screening is attainable, the welfare-maximizing approval rule is never affine. The distortion is predictable: there is a reserve report at which pretending to be the marginal type costs exactly what approval is worth. Approving every report at or above the reserve sorts agents as if their beliefs were observed, under every strictly proper score and for any belief distribution. So does a continuous rule with one corner; under strict feasibility no continuously differentiable rule does. The binding constraint is steepness, not smoothness: a slope-capped rule sorts perfectly if and only if the cap reaches the critical slope, the steepest chord of the pretending cost up to the reserve; below it the welfare loss is cubic in the shortfall. For the Brier, log and power scores, and whenever the pretending cost is convex up to the reserve, the critical slope has a closed form. The instances are AI-agent oversight and marketplace operation.
\end{abstract}

\maketitle

\section{Introduction}\label{sec:intro}

A principal who delegates a decision to a better-informed agent often asks the agent to report what it knows before deciding whether to delegate at all. Scalable oversight of AI systems is organized this way: a model reports a confidence, a strictly proper scoring rule pays it for accuracy once the outcome is observed, and an approval rule turns the same report into permission to act, as in protocols built on human preference feedback \citep{christiano2017preferences} and on debate between models \citep{irving2018debate}. Marketplace operation has the same shape, with an operator reporting the information on which an allocation and its payments are computed and earning revenue that moves with those payments. In both settings the report is paid for twice, once for being accurate and once for what it authorizes, and only the first payment is designed to reward the truth.

The two payments pull against each other, and the pull is created by the principal. Paying for accuracy makes the report informative; letting the report decide makes distorting it worthwhile. A principal that wants to separate agents worth approving from agents not worth approving must make approval depend on the report in a way that is not affine, and a non-affine dependence moves the agent's best report away from its belief. The distortion is predictable, and a principal that sets its bar at the reserve report screens exactly as if reports were truthful; how steep an approval rule must be to reach that bar is the question that organizes the paper.

\subsection{The Primitive and Its Two Instances}
\label{sec:primitive}

Both settings instantiate a single primitive, due to \citet{boutilier2012eliciting}: an expert's payoff is a proper score for its report plus the inherent utility of the decision the report triggers. In the form used here, an agent privately observes a type $p \in [0,1]$, the probability that the action it seeks to take will succeed, and sends a report $r \in [0,1]$ settled against a binary outcome. The principal pays the agent $\beta > 0$ times a strictly proper score for $r$ and grants approval with probability $q(r)$, worth $\gamma > 0$ to the agent, where the principal designs the \emph{approval rule} $q : [0,1] \to [0,1]$. Writing $G$ for the convex generator of the score and $D_G(p,r)$ for its Bregman divergence, which is the agent's expected calibration loss from reporting $r$ when the truth is $p$, the agent maximizes $\gamma\, q(r) - \beta\, D_G(p,r)$ over $r$. The principal collects $u_s$ when an approved action succeeds, $u_f$ when it fails, and $u_d$ when it withholds approval, so it wants to approve exactly the types above the break-even probability $p_{\min} = (u_d - u_f)/(u_s - u_f)$. That is the first best, the benchmark every result below is measured against.

The two instances differ in what the report is about and in who the agent is, not in this structure. The reading just given is the AI agent oversight one, developed in Section~\ref{sec:behavioral}. In marketplace operation (Section~\ref{sec:market}), the operator's private type is the true bid profile, the score is a compliance score settled on the executed profile, and the perturbation payoff weighted by $\gamma$ is the revenue the reported profile supports; \citet{loven2026trilemma} analyze the credibility constraints that bind such markets under polymatroidal feasibility, and the instance here isolates the elicitation channel within that setting.

\subsection{Contributions}
\label{sec:contributions}

The three design results below are organized by one quantity, the critical slope of the score, defined in Section~\ref{sec:behavioral}. A principal who screens on a scored report designs the distortion it must then read; a single reserve report absorbs that distortion at no welfare cost; and whether an approval rule can reach the reserve at all is settled by how steeply it is allowed to climb.

\paragraph{Optimal oversight is self-undermining.} Provided the approval prize is not so large that pretending is cheap for everyone, that is $\gamma/\beta \leq D_G(p_{\min}, 1)$, the approval rule that maximizes the principal's welfare is never affine in the report, because an affine rule spreads approval smoothly across types and so cannot deliver the all-or-nothing approval the first best calls for (Theorem~\ref{thm:optimal-nonaffine}), and a non-affine rule makes truthful reporting suboptimal for the agent. The principal's own optimization therefore creates the conditions under which the report it is optimizing over stops being calibrated, which is the endogeneity the title names.

\paragraph{The induced distortion is predictable and absorbed by the reserve.} Under the same condition there is one report $r_0(G)$, the \emph{reserve}, at which the agent's calibration cost of pretending to be the marginal type exactly exhausts the approval prize, $D_G(p_{\min}, r_0) = \gamma/\beta$; approving exactly the reports at or above $r_0(G)$ screens the types perfectly for every strictly proper score (Theorem~\ref{thm:reserve}), so the miscalibration the principal induces costs it nothing even though nothing undoes it. The reserve is a functional of the score's generator, the break-even type and the agent's price of approval $\gamma/\beta$ alone, and not of the type distribution (Proposition~\ref{prop:f-free}), which is the reverse of the comparative static that governs a \citet{myerson1981optimal} reserve price.

\paragraph{The binding constraint is steepness, not smoothness.} Besides a cliff at $r_0(G)$, an approval rule that is Lipschitz with a single kink attains the first best for every strictly proper score, while, whenever the feasibility inequality is strict and the reserve is therefore interior, no continuously differentiable rule attains it exactly (Theorem~\ref{thm:kink}, Corollary~\ref{cor:c1}). Every first-best rule must have slope, and the requirement is sharp: writing $\Lambda$ for the largest slope the rule is permitted, first-best screening is attainable if and only if $\Lambda \geq \Lambda_c(G)$, where the critical slope $\Lambda_c(G)$ is the steepest average rate at which the pretending cost $D_G(p_{\min}, \cdot)$ climbs on the way to the reserve, that is its steepest chord (Theorem~\ref{thm:critical-slope}), and below the threshold the welfare loss is of order $(1 - \Lambda/\Lambda_c)^3$ (Proposition~\ref{prop:below-threshold}).

Together the three results show that the cliff serves only to buy slope, and $\Lambda_c(G)$ measures how much slope is enough; Corollary~\ref{cor:first-best} characterizes every first-best rule on the monotone class, and Proposition~\ref{prop:lambda-sat} treats the regime in which the prize exceeds the whole calibration budget and no reserve exists.

\paragraph{The impossibility and its benchmark.} Theorem~\ref{thm:impossibility} states the title's impossibility as an endogenous one: the perturbation that breaks truthfulness is the principal's own optimum, and the conditions NT1--NT3 say when such a perturbation is binding, undetectable and therefore inescapable. Its classical core is \citet{boutilier2012eliciting}'s (Lemma~\ref{lem:perturbation}). Appendix~\ref{app:benchmark} supplies the benchmark the restricted architecture is measured against: every incentive-compatible direct mechanism in this setting is a proper score minus a charge of exactly $\gamma$ on the approval event (Theorem~\ref{thm:canonical-form}), so the extra instrument buys a truthful message and no welfare when the reserve is feasible (Theorem~\ref{thm:welfare-equivalence}).

\subsection{Related Work}
\label{sec:related}

\paragraph{The classical layer.} \citet{boutilier2012eliciting} introduced the expert whose payoff combines a proper score with the inherent utility of the induced decision; \citet{othman2010decision} showed that decision rules acting on the report destroy properness in decision markets, and \citet{chen2014eliciting} characterized the rules that preserve it, namely the completely mixed ones. Section~\ref{sec:fenchel} carries that layer as Lemma~\ref{lem:perturbation}, where it is attributed. At issue here is the layer above it, in which the principal chooses the decision rule and the rule maximizing its welfare destroys properness.

\paragraph{Oversight of AI agents.} Three adjacent streams address the same delegation problem without this structure. Scalable-oversight protocols train or select models against human or model-generated feedback \citep{christiano2017preferences,irving2018debate}, treating the reported quantity as a training signal rather than a report the agent is separately paid for. Contract-theoretic designs for delegated effort and information acquisition pay the agent on realized outcomes and ask what effort the payment induces \citep{papireddygari2022contracts,oesterheld2020minimum}. Decision-dependent elicitation keeps a forecast informative when the forecast selects the action, through decision scoring rules \citep{oesterheld2021decision}, proxy targets \citep{wang2023proxy}, or zero-sum competition between forecasters \citep{hudson2024joint}. None has the object studied here: one principal who chooses both the score that pays for accuracy and the approval rule that pays for what the report authorizes, and who therefore designs the conflict it must then screen through.

\paragraph{Designing the score rather than the decision.} A parallel literature fixes the decision and optimizes the scoring rule. \citet{li2020optimization}, with a conference abstract in \citet{li2022optimizationabstract}, locate the value-maximizing score for costly information acquisition at the degenerate end of the curvature spectrum, a prior-anchored threshold score; \citet{hartline2023optimal} extend this to multi-dimensional effort, \citet{neyman2021binary} characterize binary scores that reward precision, and \citet{johnstone2011tailored} tailor a score to a decision maker's known payoff. The design variable here is the complement of theirs: the score is fixed, any strictly proper one serves, and the design chooses the approval rule downstream of it. The difference is not cosmetic, since the score's curvature prices pretending; collapsing it is optimal when the score is the only lever, and would remove the reserve on which the oversight design rests.

\paragraph{Cheap talk, lying costs, disclosure and persuasion.} This model sits between costless unverifiable messages \citep{crawford1982strategic} and verifiable disclosure \citep{grossman1981informational,milgrom1981good,dye1985disclosure}: the agent can lie but pays for lying, at a price of $\gamma/\beta$ in units of accuracy. That is the structure of communication with lying costs \citep{kartik2009strategic,kartik2007credulity}, where a sender pays a convex cost to misreport to a receiver who infers in equilibrium and inflates its message in response; two features separate the present model from it. The receiver here commits to the decision rule, so the perturbation is a designed object rather than an equilibrium inference, and the cost of lying is the score itself, so its curvature prices the lie and enters every design result below through the reserve and the critical slope. It also reverses the timing of Bayesian persuasion \citep{gentzkow2011bayesian}, in which the sender commits before observing the state; commitment by the reporter would restore that timing, under the conditions Section~\ref{sec:resolution} states.

\paragraph{Credible mechanism design.} \citet{akbarpour2020credible} show that no auction is simultaneously strategy-proof, credible and revenue-optimal, and subsequent work relaxes each leg in turn \citep{ferreira2020credibility,li2017obviously,dworczak2020mechanism}; \citet{loven2026trilemma} carry the trilemma to polymatroidal feasible regions. The marketplace instance below applies the elicitation frame to that setting rather than strengthening those results; Remark~\ref{rem:al-relationship} states the relationship precisely.

\paragraph{Elicitation.} The characterization of strictly proper scoring rules runs from \citet{savage1971elicitation} to \citet{gneiting2007strictly}, with the convex-analytic account in \citet{schervish1989general}. That literature asks which quantities can be elicited truthfully; the question here starts after a truthful elicitation mechanism is in place and asks what a second payoff channel does to it.

\section{Model}\label{sec:model}

A scoring mechanism already pays a \emph{reporter} for the accuracy of its report, and a \emph{receiver} takes a decision on the basis of that report; the two are the agent and the principal of Section~\ref{sec:intro}, under names general enough to carry both instances. Only the decision is the receiver's instrument, because the score is inherited rather than designed.

\subsection{The Credibility Game}

The combined objective $V = \Sbar + \gamma h$ of~\eqref{eq:combined-objective} below is the primitive of \citet{boutilier2012eliciting}, who asks how a principal can restore properness of the \emph{net} payoff by adding a compensation term to the score. We keep the primitive and invert the instrument set: the receiver chooses the decision rule alone, which is the perturbation $h$. The score here is institutionally fixed, a public benchmark, a supervisory audit standard, or an already-trained reward head, maintained on a slower timescale and typically owned by a different party than the one granting the decision; and no compensation channel exists, because the reporter's stake in the decision is not a contractible transfer that can be netted out of the score.

\begin{definition}[Credibility Game]\label{def:credibility-game}
A \emph{credibility game} is a tuple $\calG = (\Theta, \calR, \Omega, \Sbar, g, h, \calI, \mu, \gamma)$ where:
\begin{itemize}
\item $\Theta \subseteq \R^d$ is a convex, compact type space with non-empty interior. The \emph{reporter} holds a private type $\theta \in \Theta$.
\item $\calR \subseteq \R^d$ is a convex report space with non-empty interior.%
\footnote{The implicit function theorem arguments use non-empty interior. The oversight instance takes $\calR_O = [0,1]$; the perturbation analysis runs on the interior $(0,1)$, and the boundary reports $r \in \{0,1\}$ are $\mu$-null under any continuous type distribution.}
\item $\Omega$ is a measurable outcome space.
\item $\Sbar: \calR \times \Theta \to \R$ is the \emph{expected score function}, derived from a scoring mechanism $S: \calR \times \Omega \to \R$ via $\Sbar(r;\theta) = \E_\theta[S(r,\omega)]$, where $\E_\theta$ denotes expectation under the outcome distribution induced by type $\theta$.
\item $g: \Theta \to \calR$ is the \emph{truthful report function}: for each $\theta \in \Theta$, $g(\theta)$ is the unique maximizer of $\Sbar(\cdot\,;\theta)$.
\item $h: \calR \to \R$ is a continuously differentiable \emph{perturbation payoff} representing the reporter's benefit from its report beyond the scoring mechanism.
\item $\calI = (\calY, \pi)$ is the \emph{information structure}, where $\calY$ is a signal space observed by the receiver and $\pi: \calR \times \Theta \to \Delta(\calY)$ specifies the conditional distribution of signals given the report and type.
\item $\mu \in \Delta(\Theta)$ is a common prior with full support, and $\gamma > 0$ is the weight the reporter places on the perturbation payoff.
\end{itemize}
The reporter's combined objective is
\begin{equation}\label{eq:combined-objective}
V(r;\theta,\gamma) = \Sbar(r;\theta) + \gamma \cdot h(r), \quad \gamma > 0.
\end{equation}
\end{definition}

The timing is: nature draws $\theta \sim \mu$; the reporter observes $\theta$ and chooses $r$ to maximize $V(r;\theta,\gamma)$; the outcome $\omega$ realizes under the distribution indexed by $\theta$; the reporter receives $S(r,\omega) + \gamma h(r)$; and the receiver observes $y \sim \pi(r,\theta)$ and updates beliefs. The prior is also the reference measure: throughout the paper, ``positive measure'' and ``almost every type'' are statements about $\mu$.

\begin{assumption}[Regularity]\label{assum:regularity}
For every $\theta \in \Theta$: (i) $g(\theta) \in \operatorname{int}\calR$; (ii) $\Sbar(\cdot\,;\theta)$ is twice continuously differentiable on $\operatorname{int}\calR$, and $-\nabla^2_r \Sbar(r;\theta)$ is positive definite at every $r$ in the convex hull of $g(\Theta)$; and (iii) either $\calR$ has compact closure, or $\Sbar(r;\theta) + \gamma h(r) \to -\infty$ as $\|r\| \to \infty$ for each $\theta$ and each $\gamma > 0$.
\end{assumption}

Part~(ii) is a restriction and not a consequence of strict properness: a strictly proper score can have a vanishing second derivative at isolated reports, and the perturbation formula of Section~\ref{sec:fenchel} fails there. In the binary-outcome case it amounts to $G \in C^2$ with $G'' > 0$ on the relevant range, which holds for the Brier \citep{brier1950verification} and logarithmic scores. Part~(iii) is an existence condition, and its first alternative covers every instance here.

\paragraph{Scope: outcome dimension and type dimension.}
Definition~\ref{def:credibility-game} is stated for $d$-dimensional types and reports. The oversight instance specializes to the \emph{binary-outcome scalar-type} case, $\Theta = \calR = [0,1]$ and $\Omega = \{0,1\}$; the marketplace instance is $n$-dimensional and uses only the impossibility half. The two restrictions have different roles. A binary outcome makes the score's generator a function of one variable, so its curvature $G''$ is a scalar, which delivers the reserve and slope results of Section~\ref{sec:behavioral} (Theorems~\ref{thm:reserve} and~\ref{thm:critical-slope}); a one-dimensional type space makes the induced screening a cutoff, and so makes the receiver's problem a threshold problem at all. The impossibility (Theorem~\ref{thm:impossibility}) requires neither restriction and is proved in $\R^d$; Appendix~\ref{sec:multi-dim} records what carries to $d$-dimensional types. The design results are scalar-type results and we claim them only there.

We write $r^*(\theta,\gamma)$ for the reporter's optimal report and call its excess over the truthful one, $r^*(\theta,\gamma) - g(\theta)$, the \emph{inflation} of type $\theta$, the object the title calls miscalibration. In the scalar oversight instance it is signed: under a non-decreasing approval rule the reporter never deflates.

\subsection{Properness, the Bregman Cost, and the Accuracy Weight}

\begin{definition}[Strict Properness]\label{def:strict-properness}
The scoring mechanism $\Sbar$ is \emph{strictly proper} if $\Sbar(g(\theta);\theta) > \Sbar(r;\theta)$ for every $\theta \in \Theta$ and every $r \neq g(\theta)$.
\end{definition}

\begin{remark}[Savage representation and the accuracy weight]\label{rem:savage-characterization}
By the classical characterization \citep{savage1971elicitation,mccarthy1956measures,gneiting2007strictly,schervish1989general}, strict properness is equivalent to the existence of a strictly convex \emph{generator} $G$ with
\[
\Sbar(r;\theta) = G(r) + \nabla G(r) \cdot (g(\theta) - r)
\]
up to functions of $\theta$ alone: the expected score is affine in the truth and strictly concave in the report, and perturbations undermine that affine dependence (Section~\ref{sec:fenchel}). The equivalence presumes $\Sbar$ is an expected score, hence affine in the outcome distribution the type induces. A reduced-form score written directly on $(r,\theta)$ can satisfy Definition~\ref{def:strict-properness} without admitting a generator, as the kinked variant discussed in Section~\ref{sec:market} does not. Equivalently, the expected score the reporter forfeits by reporting $r$ rather than the truth is the Bregman divergence of that generator, $D_G(g(\theta), r) := G(g(\theta)) - G(r) - \nabla G(r) \cdot (g(\theta) - r) > 0$ for $r \neq g(\theta)$. Where an instance weights the two channels separately we write the reporter's objective as $V(r;\theta) = -\beta D_G(g(\theta),r) + \gamma h(r)$, with $\beta > 0$ the reporter's \emph{accuracy weight}: the price it pays per unit of Bregman cost under whatever generator its score happens to have. Only the ratio $\gamma/\beta$ matters.
\end{remark}

\paragraph{Terminology: calibration.}
Throughout, \emph{calibration} is used in the elicitation sense: a report is calibrated when it equals the reporter's own conditional probability, which under a strictly proper score is exactly the report that maximizes the expected score. This is not the statistical-forecasting sense, in which a \emph{sequence} of forecasts is calibrated when outcomes occur at the announced long-run frequency. The two agree for a truthful reporter and come apart under strategic misreporting, and nothing below is a statement about the empirical calibration of a forecast sequence. For the same reason $\beta$ is the accuracy weight: it scales the generator's Bregman cost, and reading it as a sensitivity to one particular score would misstate every result that quantifies over $G$.

\subsection{Non-Trivial Structure}

\begin{definition}[Non-Trivial Structure]\label{def:nontrivial}
A credibility game $\calG$ has \emph{non-trivial structure} if:
\begin{itemize}
\item[\textbf{(NT1)}] \emph{Binding conflict.} There exists $\Cbind \subseteq \Theta$ with $\mu(\Cbind) > 0$ such that $h(g(\theta)) < \sup_r h(r)$ for all $\theta \in \Cbind$. The truthful report does not maximize the perturbation payoff.
\item[\textbf{(NT2)}] \emph{Non-affine perturbation with a live gradient.} $h$ is not affine on any open neighborhood of $\{g(\theta) : \theta \in \Cbind\}$ in $\calR$, and $\nabla h(g(\theta)) \neq 0$ for all $\theta$ in a subset of $\Cbind$ of positive measure.
\item[\textbf{(NT3)}] \emph{Undetectability.} For each $\theta \in \Cbind$, the reporter's optimal deviation $r^*(\theta,\gamma)$ is observationally equivalent to a truthful report by some other type: there exists $\theta' \in \Theta$ with $g(\theta') = r^*(\theta,\gamma)$ and $\pi(r^*(\theta,\gamma), \theta) = \pi(g(\theta'), \theta')$ almost everywhere on $\calY$, so the signal does not separate the two.
\end{itemize}
\end{definition}

Since $r^*$ depends on $\gamma$, so does NT3, which is why $\gamma$ is part of the game.

\begin{remark}[The reach and cost of NT2]\label{rem:nt2-generic}
The gradient clause of NT2 is a hypothesis, not a corollary of non-affinity, and it is the clause the impossibility consumes. The affine functions form a closed, nowhere-dense subset of $C^1(\calR)$ in the $C^1$ topology, so non-affinity of $h$ is generic, but genericity is not the whole of NT2: non-affinity is a property of $h$ alone, whereas NT2 is a joint condition on $h$ and on where the truthful reports $g(\Cbind)$ sit inside $\calR$, and a perturbation that is non-affine somewhere can still be stationary exactly on $g(\Cbind)$. The gradient clause rules that out, and it is not free: a smooth $h$ with an interior maximum at $g(\theta_0)$ satisfies non-affinity and fails the gradient clause at $\theta_0$. Such types are not immune, only slower, since Lemma~\ref{lem:perturbation}(iii) shows that each of them deviates once $\gamma$ exceeds a finite type-dependent threshold. The impossibility is therefore unconditional in $\gamma$ on the set where the gradient is live, and holds for large $\gamma$ elsewhere.
\end{remark}

\paragraph{Affine perturbations, inversion and pooling.}
An affine perturbation $h(r) = a + b^\top r$ has a live gradient and so moves every type's best report (Lemma~\ref{lem:perturbation}(ii)), but it loses no information as long as the best response stays interior. At an interior best response at which $h$ is differentiable, the first-order condition under the Savage form of Remark~\ref{rem:savage-characterization} reads $g(\theta) = r - \gamma\,[\nabla^2 G(r)]^{-1}\nabla h(r)$ whenever $\nabla^2 G(r)$ is nonsingular, so the truthful report is recovered from the report, for affine and non-affine $h$ alike, by a receiver who knows the primitives. Under constant curvature, the Brier case, the affine shift $\gamma[\nabla^2 G]^{-1} b$ is the same for every type, the perturbed mechanism is the same score with generator $G$ and truthful report $g(\theta) + \gamma[\nabla^2 G]^{-1} b$, and a single published constant restores calibration for every consumer of the report. First-best screening requires a non-affine rule (Theorem~\ref{thm:optimal-nonaffine}), and every rule that attains it sends the types between the break-even probability and the reserve to one report, a jump or a kink of the rule (Lemma~\ref{lem:structure} with Lemma~\ref{lem:envelope}). There no correction recovers the truthful report, and the report taken at face value stays miscalibrated, which costs nothing to a receiver that only screens (Theorem~\ref{thm:reserve}) and costs other consumers of the report what Section~\ref{sec:discussion} prices under the calibration-essential heading.

\paragraph{Reading NT3 across domains.}
NT3 is one formal condition, signal indistinguishability, but the friction that generates it differs by domain. Sealed-bid privacy in a marketplace and mandatory-rotation gaps in auditing are \emph{design-contingent}: an ascending format, or a shorter rotation cycle, removes them, so NT3 there flags an opportunity for redesign rather than a fundamental limit. Adverse selection in AI oversight is \emph{structural}: the agent's type reflects its internal state, and no format change makes it observable. Enriching the exogenous $\calI$ shrinks the NT3 set monotonically in the Blackwell order \citep{blackwell1951comparison}, and a fully informative signal dissolves the impossibility.

\section{Perturbation of Truthfulness}\label{sec:fenchel}

A report-dependent payoff that is not affine destroys the properness of a fixed score; the receiver's decision rule is such a payoff; therefore a receiver who screens on a scored report cannot also expect that report to be truthful. Lemma~\ref{lem:perturbation} is that statement and Theorem~\ref{thm:impossibility} is the version of it that this paper uses. Section~\ref{sec:behavioral} then makes the perturbation endogenous: the receiver's own welfare-optimal rule supplies it.

In proper scoring rules and in dominant-strategy payment rules alike, the agent's utility is affine in the type around a strictly convex potential, and truthfulness is pinned by that potential's first-order condition \citep{schervish1989general,abernethy2012characterization,vohra2011mechanism}. We organize the argument around that shared structure rather than around any one of its instances.

\begin{lemma}[Perturbation Lemma]\label{lem:perturbation}
Let $\Sbar$ be strictly proper with truthful report $g$ and let $h \in C^1(\calR)$, under Assumption~\ref{assum:regularity}.
\begin{enumerate}
\item[(i)] \emph{(Necessity of a flat perturbation.)} If $g(\theta)$ is a global maximizer of $V(\cdot\,;\theta,\gamma)$ for every $\theta \in \Theta$ and every $\gamma > 0$, then $h$ is constant on $g(\Theta)$.%
\footnote{The converse is false, and we do not use it: $h$ constant on $g(\Theta)$ still permits a profitable deviation to a report outside $g(\Theta)$ at which $h$ is larger. Truthfulness is restored by $h$ constant on all of $\calR$, which is the domain-separation pathway of Section~\ref{sec:resolution}.}

\item[(ii)] \emph{(Generic destruction.)} If $\nabla h(g(\theta)) \neq 0$ for $\theta$ in a positive-measure subset of $\Cbind$, then for every $\gamma > 0$ and every such $\theta$ the truthful report fails to maximize $V(\cdot\,;\theta,\gamma)$; in particular truthfulness fails on a positive-measure subset of $\Cbind$.

\item[(iii)] \emph{(Residual types.)} If $\nabla h(g(\theta_0)) = 0$ for some $\theta_0 \in \Cbind$ with $h(g(\theta_0)) < \sup_r h(r)$, and $h \in C^2$ near $g(\theta_0)$, then there is a type-dependent threshold $\bar\gamma(\theta_0) > 0$ such that $g(\theta_0)$ fails to maximize $V(\cdot\,;\theta_0,\gamma)$ for every $\gamma > \bar\gamma(\theta_0)$.
\end{enumerate}
Parts~(i) and~(ii) are the fixed-score, no-compensation specialization of \citet{boutilier2012eliciting}'s characterization of elicitation under decision-dependent payoffs, and they cover the ground of \citet{chen2014eliciting}'s result that strictly proper decision making requires completely mixed decision rules; see also \citet{othman2010decision}. We claim no novelty for them. Part~(iii), which locates the residual zero-gradient types and gives them a finite threshold, is a refinement of that classical statement and is proved in Appendix~\ref{app:perturbation}.
\end{lemma}

\begin{proof}[Proof sketch]
By the first-order condition of the score at the truthful report, which is interior by Assumption~\ref{assum:regularity}(i), the gradient of $V$ there is
\begin{equation}\label{eq:foc-perturbed}
\nabla_r V(g(\theta);\theta,\gamma) = \gamma \nabla h(g(\theta)),
\end{equation}
which is nonzero under the hypothesis of~(ii), so $g(\theta)$ is not a critical point of $V$. Part~(i) is the contrapositive: a type whose truthful report earns less $h$ than another type's deviates to the latter once $\gamma$ outweighs the fixed score loss. Part~(iii) is a second-order argument at a zero-gradient type, proved in Appendix~\ref{app:perturbation}.
\end{proof}

\subsection{The Credibility Impossibility}

The theorem below states the impossibility and, under a little more regularity, the size of the deviation it forces: truthfulness and rationality cannot hold together, and the rational report moves away from the truth by an amount proportional to the perturbation's gradient and inversely proportional to the score's curvature.

\begin{theorem}[Credibility Impossibility]\label{thm:impossibility}
Let $\calG$ be a credibility game satisfying Assumption~\ref{assum:regularity} and conditions NT1 and NT2 of Definition~\ref{def:nontrivial}. Then no strategy $\sigma: \Theta \to \calR$ simultaneously achieves:
\begin{itemize}
\item[\textbf{(T)}] Truthfulness: $\sigma(\theta) = g(\theta)$ for almost every $\theta$.
\item[\textbf{(R)}] Rationality: $\sigma(\theta) \in \arg\max_{r \in \calR} V(r;\theta,\gamma)$ for some $\gamma > 0$.
\end{itemize}
Assume in addition that $\calR$ is compact, that $h \in C^2$, and that $\Sbar(\cdot\,;\theta)$ is continuous on $\calR$ and $C^3$ near $g(\Theta)$; the conjunction above needs none of these. Then, as $\gamma \to 0$, the rational reporter's optimal report satisfies the \emph{perturbation formula}
\begin{equation}\label{eq:perturbation-formula}
r^*(\theta,\gamma) = g(\theta) + \gamma \cdot
\bigl[-\nabla^2_r \Sbar(g(\theta);\theta)\bigr]^{-1} \nabla h(g(\theta))
+ O(\gamma^2),
\end{equation}
where $-\nabla^2_r \Sbar(g(\theta);\theta)$ is positive definite by Assumption~\ref{assum:regularity}(ii), and the expected score forfeited by that deviation is
\begin{equation}\label{eq:scoring-loss}
\Sbar(g(\theta);\theta) - \Sbar(r^*(\theta,\gamma);\theta) =
\frac{\gamma^2}{2}\, \nabla h(g(\theta))^\top
\bigl[-\nabla^2_r \Sbar(g(\theta);\theta)\bigr]^{-1} \nabla h(g(\theta))
+ O(\gamma^3),
\end{equation}
quadratic in $\gamma$ and scaling with the inverse curvature of the score. The first display is the implicit-function comparative static of the maximizer, not an envelope term; the second is the second-order term of the value function $\gamma \mapsto \max_r V(r;\theta,\gamma)$, whose first-order term is the envelope term $h(g(\theta))$ \citep{milgrom2002envelope}.
\end{theorem}

\begin{proof}
The impossibility of (T) $\wedge$ (R) on $\Cbind$ is Lemma~\ref{lem:perturbation}(ii): NT1 makes $\Cbind$ a set of positive measure, and NT2 supplies a positive-measure subset of it with $\nabla h(g(\theta)) \neq 0$, at which the truthful report is not a critical point of $V$ by~\eqref{eq:foc-perturbed}, hence not a best response. For the formula, the interior first-order condition is $\nabla_r \Sbar(r;\theta) + \gamma \nabla h(r) = 0$, solved at $\gamma = 0$ by $r = g(\theta)$, which is interior by Assumption~\ref{assum:regularity}(i). The Jacobian $\nabla^2_r \Sbar(g(\theta);\theta)$ is negative definite and hence invertible, so the implicit function theorem gives a smooth branch $r^*(\theta,\cdot)$ through $(g(\theta),0)$; differentiating at $\gamma = 0$ yields~\eqref{eq:perturbation-formula}. The branch is the global maximizer for $\gamma$ below a type-dependent threshold: $V(\cdot\,;\theta,\gamma)$ is continuous on the compact $\calR$, so a global maximizer exists, every limit point of global maximizers as $\gamma \to 0$ maximizes $\Sbar(\cdot\,;\theta)$ by Berge's theorem and so equals $g(\theta)$ by strict properness, and for small $\gamma$ the global maximizer therefore lies in the neighbourhood on which the branch is the unique critical point. The expansion is asymptotic in $\gamma$ and is not claimed uniformly in $\theta$. Substituting $r^* - g(\theta)$ from~\eqref{eq:perturbation-formula} into the second-order Taylor expansion of $\Sbar$ at $g(\theta)$ gives~\eqref{eq:scoring-loss}.
\end{proof}

\paragraph{Equilibrium concept and the role of NT3.}
Theorem~\ref{thm:impossibility} is an \emph{ex post} statement, hence stronger than the corresponding Bayesian Nash claim and independent of the prior, and NT3 plays no part in it. NT3 carries it into an equilibrium prediction: without NT3 the deviation is visible and a receiver who sees it can punish it, and with NT3 it is statistically indistinguishable from truthful reporting by some other type within the monitoring horizon. Two consequences follow. Existence of a deviation is enough, so the conclusion holds in \emph{every} equilibrium in which the reporter best-responds to the combined objective; and the formal results of Sections~\ref{sec:fenchel} and~\ref{sec:behavioral} do not use the information structure $\calI$ at all, since $\calI$ enters only through NT3's reading, which Section~\ref{sec:behavioral} makes quantitative and Section~\ref{sec:market} invokes only to license the word impossibility. We reserve the unqualified phrase ``the impossibility'' for the full NT1--NT3 statement.

\subsection{Resolution Mechanisms}\label{sec:resolution}

Four interventions bear on the impossibility and only one of them resolves it: domain separation deletes a hypothesis of Theorem~\ref{thm:impossibility} and restores truthfulness outright, enforced commitment and competition change what the receiver can enforce or learn without altering the reporter's best response in the game as specified, and bounding the gain prices the deviation instead. These are intervention points rather than results of this paper, and we claim none of them.

\paragraph{Enforced commitment (adds an enforcement technology outside NT1--NT3).} Suppose the receiver fixes the decision rule and the commitment protocol before the reporter observes its type, and an enforcement technology binds the reporter to a mandated truthful strategy $\hat{r} = g$, any deviation from it being blocked or voided by a verifier. Enforcement is an institutional endowment and not a consequence of the payoff~\eqref{eq:combined-objective}: it is assumed, not derived from detection, and it is not an equilibrium outcome of the game as specified. Commitment by the reporter itself would not help: $\E_\mu[V(\hat{r}(\theta);\theta,\gamma)]$ is separable in $\theta$, so it is maximized pointwise at $\hat{r}(\theta) = r^*(\theta,\gamma)$, the deviation the theorem produces. Committing the reporter also reverses the informational timing of the main game, in which the receiver moves first, and turns it into a \citet{gentzkow2011bayesian} persuasion problem with the reporter as sender and the approval decision as the receiver's action, when the receiver's payoff depends on $\theta$ only through the posterior mean and reports and signals share a space, as in the oversight instance.

\paragraph{Domain separation (eliminates NT1, and NT2 with it).} If the perturbation payoff is constant on the whole report space, $h \equiv c$, the binding conflict vanishes and truthful reporting is restored by strict properness alone. This is exact, and it is also unavailable to a receiver who wants to screen, since a constant $h$ is a decision rule that ignores the report.

\paragraph{Competition (improves detection; weakens no hypothesis of the impossibility).} If $n \geq 2$ reporters with correlated information report simultaneously, the receiver sees a cross-section whose \emph{location} is uninformative, since a common inflation shifts all $n$ reports together and mimics truthful reporting of a shifted state, but whose \emph{shape} is not: the deviating types pool on the reserve $r_0$ (Theorem~\ref{thm:reserve}), so the cross-section carries an atom no truthful reporting distribution with a density can produce, and a goodness-of-fit test at a fixed level rejects with probability tending to $1$ as $n \to \infty$. Three hypotheses the credibility game does not supply carry that claim: types conditionally independent given the state with a common continuous truthful law; a receiver that tests goodness of fit against that law rather than a coarser statistic; and a channel from detection to payoff, which Definition~\ref{def:credibility-game} does not provide at all, so detection moves the receiver's posterior and nothing else. Competition can also enable ratings shopping \citep{skreta2009ratings,becker2011credit}; the multi-sender benchmark is \citet{gentzkow2017competition}, and the competition \citet{loven2026trilemma} find orthogonal to execution credibility is between the institutions that execute, not the cross-comparison used here.

\paragraph{Bounding the gain through the scoring family (orthogonal to NT1--NT3).} The impossibility holds for every strictly proper score satisfying Assumption~\ref{assum:regularity}, so no choice inside that class escapes it and NT1--NT3 are untouched; a strictly proper score kinked at the truthful report falls outside that class and can escape it (Section~\ref{sec:market}). The choice can nonetheless bound the \emph{magnitude} of the equilibrium deviation, which enters the welfare loss through~\eqref{eq:scoring-loss}. This pathway prices the endogeneity; it does not remove it.

\section{The Oversight Design Problem}\label{sec:behavioral}

Section~\ref{sec:fenchel} took the perturbation as given and measured what it costs. This section asks what a principal does who knows about the perturbation and designs against it. Two modeling choices, neither of them a normalization, shape every result below.

\emph{The instrument is an approval probability, not a transfer.} In a deployment gate on an autonomous agent, or a regulator's authorization of a self-reported risk assessment, the score is fixed upstream by a training pipeline or a disclosure standard, and the one lever left is whether the reported action is executed autonomously.

\emph{The agent maximizes globally.} Given $q$, the agent compares all reports and takes a global maximizer, not an interior stationary point. The two conventions select different behavior in the slope-constrained regime of Proposition~\ref{prop:below-threshold}, where the agent's optimum does become interior.

Four objects carry the section, each introduced where it is first used: the calibration cost $C$ of the marginal type and the reserve $r_0$, the report at which that cost equals the approval prize (Section~\ref{sec:game}); the Bregman ramp, the rule that pays exactly that cost (Section~\ref{sec:reserve}); and the critical slope $\Lambda_c$, the least steepness any rule needs (Section~\ref{sec:critical}). One running example under the Brier score accompanies the results, and Figure~\ref{fig:rules} draws its rules and the reports they induce.

\subsection{The Oversight Game}\label{sec:game}

\begin{definition}[Oversight Game]\label{def:oversight-game}
An \emph{oversight game} is a tuple
\[
\calG_O = (\Theta_O, \calR_O, \Omega_O, G, q, F, u_s, u_f, u_d, \beta, \gamma)
\]
with type space $\Theta_O = [0,1]$, the type $p$ being the true success probability of the agent's preferred action $a^*(c)$ in context $c$; report space $\calR_O = [0,1]$, the agent reporting confidence $r$; binary outcome space $\Omega_O = \{0,1\}$; a strictly convex generator $G$ for the strictly proper score, whose Bregman divergence $D_G$ of Remark~\ref{rem:savage-characterization} is, in this binary case, $D_G(p,r) = G(p) - G(r) - G'(r)(p-r)$, the agent's expected calibration loss from reporting $r$ when the truth is $p$, the Brier score being the case $G(p) = p^2$ with $D_G(p,r) = (r-p)^2$; the principal's \emph{approval rule} $q : [0,1] \to [0,1]$, which grants autonomous execution with probability $q(r)$; a type distribution $F$ with continuous density $f > 0$ on $[0,1]$; the principal's payoffs $u_s > u_d > u_f$ from success, delegation and failure; the agent's accuracy weight $\beta > 0$ of Remark~\ref{rem:savage-characterization}; and its autonomy weight $\gamma > 0$, the value of receiving approval. The \emph{first-best threshold} is $p_{\min} = (u_d - u_f)/(u_s - u_f)$, the success probability at which the principal is indifferent between approving and delegating.
\end{definition}

The principal commits to $q$ first (Remark~\ref{rem:timing-robustness} records what changes without that timing); the agent then observes context $c$ and type $p$ and reports confidence $r$, and the principal approves autonomous execution with probability $q(r)$. The agent's combined objective is
\begin{equation}\label{eq:behavioral-objective}
V(r;p) = -\beta\, D_G(p,r) + \gamma\, q(r),
\end{equation}
which specializes to $-\beta(r-p)^2 + \gamma q(r)$ under the Brier score. The agent plays a global best response $r^*(p) \in \arg\max_{r \in [0,1]} V(r;p)$, and the \emph{induced screening function} $\tau(p) \coloneqq q(r^*(p))$ is the approval probability that type $p$ actually receives. Two conventions are in force throughout. Every rule considered is one for which that argmax is non-empty at every type, and ties inside it are broken in the principal's favour, so that $\tau(p) = \max\{q(r) : r \in \arg\max_{r \in [0,1]} V(r;p)\}$ and welfare below is welfare under that selection. The first-best results hold under any selection, because the best response they induce is unique off an $F$-null set; the selection is used only where welfare is compared across rules, in Theorem~\ref{thm:critical-slope}(ii) and Proposition~\ref{prop:below-threshold}.

Write $W_A(p) = p u_s + (1-p) u_f$ for the principal's payoff from approving type $p$, $A \coloneqq u_s - u_f > 0$ for the stake of the decision, and $\Pi(p) = W_A(p) - u_d = A(p - p_{\min})$ for the net gain from approving. The principal's expected utility under $q$ is
\begin{equation}\label{eq:principal-utility}
W(q) = \int_0^1 \bigl[\tau(p)\, W_A(p) + (1-\tau(p))\, u_d\bigr] f(p)\, dp = u_d + \int_0^1 \tau(p)\, \Pi(p)\, f(p)\, dp,
\end{equation}
and since $\Pi$ changes sign exactly at $p_{\min}$, the pointwise optimum is $\tau(p) = \mathbf 1\{p \ge p_{\min}\}$ and the first best is
\begin{equation}\label{eq:first-best}
W^* = u_d + \int_{p_{\min}}^1 \Pi(p) f(p)\, dp .
\end{equation}
Subtracting~\eqref{eq:principal-utility} from~\eqref{eq:first-best} gives the identity behind every welfare statement in this section:
\begin{equation}\label{eq:welfare-gap-identity}
W^* - W(q) \;=\; A\int_0^1 |p - p_{\min}|\, e(p)\, f(p)\, dp, \qquad e(p) \coloneqq \begin{cases} \tau(p), & p < p_{\min} \quad (\text{false approval}),\\ 1 - \tau(p), & p > p_{\min} \quad (\text{false rejection}).\end{cases}
\end{equation}
The welfare gap is a weighted $L^1$ norm of the screening error, with a weight that vanishes linearly at $p_{\min}$. A rule attains \emph{first-best screening} if and only if $e = 0$ $F$-almost everywhere.

\paragraph{Standing hypothesis.}
Throughout this section the generator $G$ is continuous and strictly convex on $[0,1]$, twice continuously differentiable on the \emph{open} interval $(0,1)$ with $G'' > 0$ there, and the prize is \emph{strictly feasible}, $\gamma/\beta < D_G(p_{\min},1)$, where $D_G$ is extended to endpoint reports by $D_G(p,0) \coloneqq \lim_{r \downarrow 0} D_G(p,r)$ and $D_G(p,1) \coloneqq \lim_{r \uparrow 1} D_G(p,r)$, so that $D_G : [0,1]^2 \to [0,+\infty]$ is finite and jointly continuous on $[0,1] \times (0,1)$ and $r \mapsto D_G(p,r)$ is continuous on all of $[0,1]$ in the extended sense. Call this hypothesis (S). Every standard score satisfies it: the Brier and power scores at every prize below $D_G(p_{\min},1)$, and the log score at every prize, since for it $D_G(p_{\min},1) = +\infty$ (Example~\ref{ex:log-score}); the open interval admits the log score, whose $G'$ diverges at both endpoints. Under (S), $D_G(p,r) > 0$ for $p \ne r$, the map $r \mapsto D_G(p,r)$ is strictly increasing on $[p,1]$ and strictly decreasing on $[0,p]$, and for every interior report $r$ the map $p \mapsto D_G(p,r)$ is strictly decreasing on $[0,r]$ and strictly increasing on $[r,1]$, all comparisons running in $[0,+\infty]$. Several results hold under less, and the weaker grades are named for that purpose, the generator and the prize graded separately. On the generator, \textbf{(A1)} is $G$ continuous, strictly convex and continuously differentiable on $(0,1)$ with the same endpoint extension; \textbf{(A2)} adds $G \in C^2((0,1))$ with $G'' > 0$; \textbf{(A3)} adds $G \in C^3((0,1))$, used only below the critical slope. On the prize, \emph{weak feasibility} is $\gamma/\beta \le D_G(p_{\min},1)$, which admits the boundary prize at which the reserve is $r_0 = 1$; so (S) is (A2) with strict feasibility. A sentence after each such result says what suffices, and a table in \suppsec{S.10} records, for every result of this section, the weakest hypothesis its proof uses.

Write $C(r) \coloneqq D_G(p_{\min},r)$ for the calibration cost borne by the marginal type in reporting $r$, so that $C(p_{\min}) = 0$, $C$ is continuous and strictly increasing on $[p_{\min},1]$ as a map into $[0,+\infty]$, and $C'(r) = G''(r)(r - p_{\min})$. The \emph{reserve} $r_0 = r_0(G)$ is the unique solution above $p_{\min}$ of $C(r_0) = \gamma/\beta$, which under (S) lies in the open interval $(p_{\min},1)$ and under weak feasibility in $(p_{\min},1]$, and $m \coloneqq r_0 - p_{\min} > 0$ is the induced miscalibration.

\paragraph{Running example.} Under the Brier score, $G(p) = p^2$, with $p_{\min} = 0.4$ and $\gamma/\beta = 0.09$: $C(r) = (r - 0.4)^2$, the reserve is $r_0 = 0.7$, the induced miscalibration is $m = 0.3$, the Bregman ramp of Section~\ref{sec:reserve} is $(r-0.4)^2/0.09$ on $[0.4, 0.7]$, and the critical slope of Section~\ref{sec:critical} is $\Lambda_c = 2/m = 20/3$. Figure~\ref{fig:rules} draws the rules of this section at this instance and the reports they induce.

\begin{figure}[t]
\centering
\includegraphics[width=\linewidth]{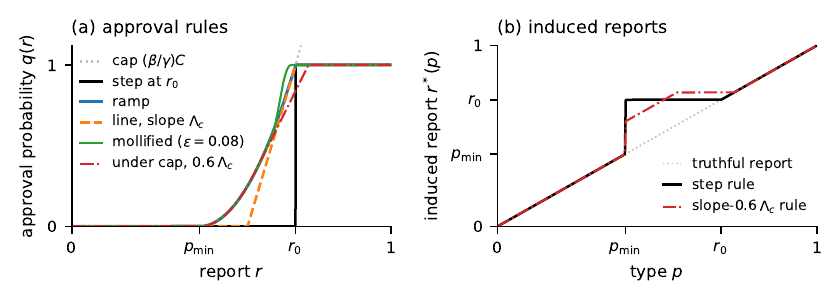}
\caption{The running example, $p_{\min} = 0.4$, $\gamma/\beta = 0.09$, $r_0 = 0.7$ under the Brier score. (a) Approval rules under the cap $(\beta/\gamma)C(r)$: the step rule, the Bregman ramp, the truncated line of slope $\Lambda_c = 20/3$, a mollified ramp ($\varepsilon = 0.08$), and the largest rule under the cap with slope $0.6\,\Lambda_c$. (b) The report each type sends under the step rule, shared by every first-best rule for all types but $p_{\min}$, and under the slope-$0.6\,\Lambda_c$ rule.}
\label{fig:rules}
\end{figure}

\paragraph{NT conditions.}
Write $r_{\min}$ for the approval threshold, the minimum report at which $q(r) > 0$; its optimal value $r_0$ is derived in Theorem~\ref{thm:reserve}. NT1 holds for the objective~\eqref{eq:behavioral-objective} because there is a set $\Cbind$ of types of positive measure with $p < r_{\min}$, as in Definition~\ref{def:nontrivial}, on which $h(g(p)) = q(p) = 0 < 1 = \sup_r q(r)$; NT2 is checked rule-specifically, because Theorem~\ref{thm:optimal-nonaffine} delivers global non-affinity of $q$ but not the gradient clause, as Remark~\ref{rem:nt2-generic} says: for the Bregman ramp $q_G$ of Theorem~\ref{thm:kink} with $\Cbind = (p_{\min},r_0)$, $q_G$ is affine on no open neighborhood of $\Cbind$, being $0$ below $p_{\min}$ and $1$ above $r_0$ with $C$ strictly increasing between, and $q_G' = (\beta/\gamma)G''(r)(r - p_{\min}) > 0$ throughout $\Cbind$, which is the gradient clause in the form Lemma~\ref{lem:perturbation} uses; and NT3 holds because $p$ is private, so an inflated report $r' > p$ is consistent with a more confident agent of type $p' = r'$; quantitatively, detecting an inflation of magnitude $\Delta$ from realized outcomes takes $\Theta(1/\Delta^2)$ observations, Hoeffding's inequality \citep{hoeffding1963probability} giving the upper bound and a two-point minimax argument \citep{loven2026behavioral} the matching lower bound, so the smaller the induced distortion, the longer the monitoring horizon needed to see it.

\subsection{Optimal Oversight Is Non-Affine}

The principal designs $q$ anticipating the agent's response, and every design that screens correctly destroys properness.

\begin{theorem}[Optimal Oversight Non-Affinity]\label{thm:optimal-nonaffine}
Let the principal choose an approval rule $q : [0,1] \to [0,1]$ to maximize expected utility under Stackelberg timing, given that the agent best-responds, and suppose the type distribution $F$ places positive mass on both sides of the first-best threshold $p_{\min}$. Then no affine $q$ is optimal for the principal.
\end{theorem}

(A1) and weak feasibility suffice, and feasibility is used only to pass from screening to optimality: that no affine $q$ attains first-best screening holds in every regime. The proof is in Appendix~\ref{app:optimal-oversight}. Attainment is Theorem~\ref{thm:reserve}, so the principal's optimum exists, equals the first best, and is not affine; above the feasibility bound no deterministic rule screens, and Proposition~\ref{prop:lambda-sat} gives the second best.

\emph{Intuition.} The principal wants a discontinuous object in type space, approve above $p_{\min}$ and reject below, out of a continuous instrument in report space. An affine rule turns a translation of reports into a translation of approvals, so whatever inflation it induces it induces uniformly, and a uniform shift cannot produce a threshold. Any rule that does is non-affine, and by Lemma~\ref{lem:structure} every rule that screens correctly pays nothing below $p_{\min}$ and full approval only from the reserve on, so every type between the two must inflate. The optimum is self-undermining because screening and properness require incompatible shapes of the instrument.

\begin{remark}[Robustness to timing]\label{rem:timing-robustness}
The pointwise-optimal induced screening $\tau(p) = \mathbf 1\{p \ge p_{\min}\}$ does not depend on timing, because $\Pi$ changes sign at $p_{\min}$ independently of the agent's strategy, and implementing it requires a non-affine $q$ by the regime-free screening half of Theorem~\ref{thm:optimal-nonaffine}. Non-affinity therefore holds in any equilibrium of the oversight game, under simultaneous moves as well as Stackelberg; only the reserve's location $r_0(G)$ is Stackelberg-specific. The conclusion is conditional on equilibrium existence, which under simultaneous timing needs regularity conditions that Stackelberg timing supplies automatically.
\end{remark}

\subsection{The Reserve Is a Generator Invariant}\label{sec:reserve}

The threshold $p_{\min} + \sqrt{\gamma/\beta}$ is a Brier artefact. The general object is the report at which the marginal type's calibration cost exactly exhausts the approval prize, and the next theorem says that approving at or above it sorts every type correctly under every strictly proper score: the principal needs the score and the prize, and nothing about the type distribution. Each attainment result of this section is a hypothesis check against the following lemma. Write $\mathrm{ramp}(r) \coloneqq 0$ for $r \le p_{\min}$ and $\min\{1, (\beta/\gamma)C(r)\}$ for $r \ge p_{\min}$ for the \emph{Bregman ramp}, the rule that pays the marginal type exactly its calibration cost up to the reserve and $1$ from the reserve on.

\begin{lemma}[Envelope sufficiency]\label{lem:envelope}
Let $q : [0,1] \to [0,1]$ lie under the ramp up to the reserve and agree with it from the reserve on, that is $q \le \mathrm{ramp}$ on $[0,r_0]$ and $q = 1$ on $[r_0,1]$. Then $q$ induces $\tau(p) = \mathbf 1\{p \ge p_{\min}\}$ for every $p \ne p_{\min}$, with the unique best response $r^*(p) = p$ for $p < p_{\min}$ and $r^*(p) = \max(p,r_0)$ for $p > p_{\min}$; at $p = p_{\min}$ the argmax is $\{p_{\min}\} \cup \{r \in [p_{\min},r_0] : q(r) = \mathrm{ramp}(r)\}$, a single ($F$-null) type.
\end{lemma}

The lemma assumes nothing about $q$ beyond the two inequalities, neither monotonicity nor continuity, and it holds under (A1) and weak feasibility alone; the supplementary material extends it to the saturated regime of Section~\ref{sec:saturated}, with $r_0$ replaced by $1$ and the value $1$ by $\lambda_{\mathrm{sat}}$. The terminal agreement cannot be dropped: a rule reaching the ramp at $r_0$ but falling away above it does not induce that screening, since the types above the reserve for which the recovered approval exceeds the cost of the lie prefer dipping to the reserve over reporting truthfully. Proof in Appendix~\ref{app:envelope}.

\begin{theorem}[The reserve rule attains first best]\label{thm:reserve}
Let $r_0$ be the reserve. Then:
\begin{enumerate}
\item[(i)] The step rule $q(r) = \mathbf 1\{r \ge r_0\}$ induces $\tau(p) = \mathbf 1\{p \ge p_{\min}\}$ for every $p \ne p_{\min}$, with strict optimality of the induced report; at $p = p_{\min}$ the agent is exactly indifferent between $r = p_{\min}$ and $r = r_0$, a single ($F$-null) type. Hence the step rule attains first-best welfare under any tie-breaking.
\item[(ii)] The reserve is an affine function of $\sqrt{\gamma/\beta}$, that is $r_0 = p_{\min} + c_0\sqrt{\gamma/\beta}$ for a constant $c_0$ independent of $(p_{\min},\gamma/\beta)$ across the whole feasible range, if and only if $G$ is quadratic, and $c_0 = 1$ if and only if $G'' \equiv 2$, the Brier normalization.
\end{enumerate}
\end{theorem}

Both parts hold under (A1) and weak feasibility, with no assumption on $F$ beyond a density, so the log score is covered (Example~\ref{ex:log-score}). Part~(ii) shows that the Brier formula is misleading rather than merely special: at $G = p^2 + \tfrac12 p^3$, $p_{\min} = 0.4$, $\gamma/\beta = 0.09$, the true reserve is $r_0 = 0.62224$ against the Brier formula's $0.7$, and the wrong threshold costs $3.59 \times 10^{-3}$ of welfare (Table~\ref{tab:numerics}, row 1).

\begin{table}[t]
\caption{Numerical illustrations printed in this section, with their reproduction parameters. The certificate archive cited at the end of the paper recomputes every row; the constants of rows 1, 3 and 4 and the slope premiums of row 5 are enclosed in interval arithmetic there; the grid quantities of rows 2 and 3 and the constant ratios of row 5 are checked in floating point.}
\label{tab:numerics}
\scriptsize
\setlength{\tabcolsep}{3.5pt}
\begin{tabular}{@{}cp{3.3cm}p{3.2cm}p{2.5cm}p{3.4cm}@{}}
\toprule
Row & Instance & Quantity & Value & Reproduction \\
\midrule
1 & $G = p^2 + \tfrac12 p^3$, $p_{\min} = 0.4$, $\gamma/\beta = 0.09$ & reserve; welfare cost of using the Brier threshold $0.7$ instead & $r_0 = 0.62224$; $3.59 \times 10^{-3}$ & $A = 1$, $F$ uniform on $[0,1]$; $r_0$ by bisection on $D_G(p_{\min},r_0) = \gamma/\beta$; gap in closed form from~\eqref{eq:welfare-gap-identity} \\
2 & Brier, $p_{\min} = 0.4$, $\gamma/\beta = 0.09$, $\varepsilon = 0.02$, bump $\varphi(x) = e^{-1/x}/(e^{-1/x} + e^{-1/(1-x)})$ & width of the falsely approved interval below $p_{\min}$; $\min_{p > p_{\min}} \tau$ & $5.5 \times 10^{-3}$; $0.9862$ & $4001$-point type grid; the mollified ramp $q_\varepsilon$ of Appendix~\ref{app:kink} \\
3 & $G''(r) = 2 + 100\exp(-((r-0.45)/0.02)^2)$, $p_{\min} = 0.4$, $\gamma/\beta = 0.2$ & $r_0$; the ramp's slope at the reserve $s(r_0)$; $\Lambda_c$; their ratio; anchor $a^\star$ of the binding chord in~\eqref{eq:critical-slope-def}; under the truncated line of slope $1.02\,s(r_0)$: $q(0)$, falsely approved mass, total mis-screened mass, worst pointwise error & $0.55084$; $1.5084$; $7.9902$; $5.30$; $0.4326$; $0.1525$; $0.40$; $0.46$; $0.80$ & $F$ uniform; $G$, $G'$ by trapezoidal quadrature of $G''$; global argmax on a $200{,}001$-point report grid over $999$ types on $[0.001, 0.999]$; the instance of \suppsec{S.9} \\
4 & log score, $p_{\min} = 0.6$, $\gamma/\beta = 0.04$ & $r_0$; $\Lambda_c$ & $0.730683$; $16.602$ & $\mathrm{KL}(p_{\min}\,\|\,r_0) = \gamma/\beta$ by bisection; $\Lambda_c = (\beta/\gamma)(r_0 - p_{\min})/(r_0(1-r_0))$ \\
5 & log score and $G = p^4$ against Brier, $p_{\min} = 0.4$, $r_0 = 0.7$ & slope premium over Brier; ratio of the upper constant to Brier's & $11.6\%$, $34.2\%$; $0.710$, $0.723$ & the closed forms~\eqref{eq:tilt-identity} and Proposition~\ref{prop:score-selection-full}(iii) at fixed $m$ \\
\bottomrule
\end{tabular}
\end{table}

\emph{Intuition.} The reserve is the report at which the marginal type is exactly indifferent between telling the truth and paying the calibration cost of climbing to approval, and that indifference is fixed by three primitives the principal does not choose when it chooses the rule: its own break-even ($p_{\min}$), the agent's price of approval ($\gamma/\beta$), and how expensive the score makes a lie of a given size ($D_G$). Lemma~\ref{lem:structure} makes the invariance exact: every rule that screens correctly saturates at this same report. A more risk-averse principal therefore needs a higher reserve, and under Brier not a steeper rule: if $u_s$ and $u_f$ both fall in risk aversion relative to $u_d$ then $p_{\min}$ increases, and since $\partial_p D_G(p,r) = G'(p)-G'(r) < 0$ for $p < r$, raising $p_{\min}$ lowers $C$ pointwise and so raises $r_0$, while under Brier $m = \sqrt{\gamma/\beta}$ and hence $\Lambda_c = 2/m$ are unchanged.

\subsection{The Structure of First-Best Rules}\label{sec:structure}

Lemma~\ref{lem:envelope} says that lying under the ramp and agreeing with it from the reserve on is sufficient for first-best screening. The next lemma says it is also necessary: every rule that sorts correctly pays at most the marginal type's calibration cost, nothing below the threshold, and full approval first at the reserve.

\begin{lemma}[Every first-best rule is capped by the ramp and saturates at the reserve]\label{lem:structure}
Let $q : [0,1] \to [0,1]$ be measurable, admit a global best response at every type and attain first-best screening. Then, on $(0,1]$, and on all of $[0,1]$ unless $G'(0^+) = -\infty$: $q \le \mathrm{ramp}$, so $q \equiv 0$ on $(0,p_{\min}]$ and $q < 1$ on $[p_{\min},r_0)$; and $\inf\{r > 0 : q(r) = 1\} = r_0$, with $q(r_0) = 1$ whenever $q$ is continuous.
\end{lemma}

Both statements hold under (A1) and weak feasibility. The cap needs only the weaker premise that $\tau(p) = 0$ for $F$-almost every $p < p_{\min}$, which is the form the saturated regime of Section~\ref{sec:saturated} uses. The report $0$ is excepted only when $G'(0^+) = -\infty$, as for the log score; it is then reachable by the $F$-null type $p = 0$ alone and carries no welfare. Proof in Appendix~\ref{app:structure}.

\emph{Intuition.} Two facts pin every correct rule and leave its shape free. It can never pay more than the marginal type's calibration cost, or types below $p_{\min}$ would buy approval; and it must reach full approval exactly at $r_0$, because the types just above $p_{\min}$ are approved and can only have climbed to the cheapest fully approved report.

\subsection{A Continuous Implementation: The Bregman Ramp}\label{sec:ramp}

The rule of Theorem~\ref{thm:reserve} is a cliff, which no approval process can implement. The screening does not require one: the next theorem says that a continuous rule with a single kink, the ramp itself, sorts every type as the cliff does, and that rounding the kink costs arbitrarily little. Write $s(r) \coloneqq (\beta/\gamma) G''(r)(r - p_{\min})$ for the ramp's slope below the reserve.

\begin{theorem}[The ramp attains first best]\label{thm:kink}
Let $q_G \coloneqq \mathrm{ramp}$, i.e.\ $q_G(r) = 0$ on $[0,p_{\min}]$, $q_G(r) = (\beta/\gamma) D_G(p_{\min},r)$ on $[p_{\min},r_0]$, and $q_G(r) = 1$ on $[r_0,1]$. Then:
\begin{enumerate}
\item[(i)] \emph{(Attainment in $C^{0,1}$.)} $q_G$ is non-decreasing and $[0,1]$-valued, is Lipschitz with constant $\sup_{[p_{\min},r_0]}s$, and induces $\tau(p) = \mathbf 1\{p \ge p_{\min}\}$ for every $p \ne p_{\min}$, with a unique best response. At $p = p_{\min}$ the argmax is the whole interval $[p_{\min},r_0]$ ($F$-null). Hence $q_G$ attains first-best welfare, and induces the same outcome as the step rule of Theorem~\ref{thm:reserve} for every type except $p_{\min}$.
\item[(ii)] \emph{(The $C^1$ infimum gap is zero.)} There is an explicit non-decreasing family $(q_\varepsilon)_{\varepsilon \in (0,m)} \subset C^1([0,1],[0,1])$, with $\mathrm{Lip}(q_\varepsilon) \le 3\sup_{[p_{\min},r_0]}s$ uniformly in $\varepsilon$, for which $W^* - W(q_\varepsilon) = O(\varepsilon)$: the $O(\varepsilon)$ term prices under-approval on $(p_{\min},r_0)$ and the false-approval channel below $p_{\min}$ enters only at $O(\varepsilon^2)$, with both constants given in Appendix~\ref{app:kink}. Hence $\inf_{q \in C^1}\bigl(W^* - W(q)\bigr) = 0$, and if $G \in C^\infty$ the family may be taken in $C^\infty$.
\end{enumerate}
\end{theorem}

Attainment in part~(i) holds under (A1) and weak feasibility; the Lipschitz constant and part~(ii) use (S). Whether the infimum in part~(ii) is attained is settled, in the negative, by Corollary~\ref{cor:c1} of the next subsection.

\emph{Intuition.} Under $q_G$ every approved type either pools at $r_0$ or reports truthfully, and every rejected type reports truthfully: the same induced outcome as the cliff, for every type but one. The cliff was therefore never necessary: its discontinuity was supplying \emph{slope}, which a Lipschitz ramp supplies continuously, and the induced miscalibration, pooling at $r_0$, is not removed by continuity. Part~(ii) says that rounding the kink costs nothing in the limit, so $C^{0,1}$ versus $C^1$ is a distinction between a maximum and a supremum.\footnote{Part~(ii) also forecloses any lower bound on the $C^1$ welfare gap: the infimum over $C^1$ rules is zero for every admissible generator, so no positive bound of the form $\mathrm{Var}_{F|\Cbind}(1/G'')\,(\gamma/\beta)^2$ can hold over that class. The obstruction is slope rather than regularity (Theorem~\ref{thm:critical-slope}), and Brier's distinguished role survives in the identity~\eqref{eq:tilt-identity}, as the minimizer of required steepness.} At the running example with $\varepsilon = 0.02$ the falsely approved interval below $p_{\min}$ has width $5.5 \times 10^{-3}$, while every type in $(p_{\min},r_0)$, an interval of width $m = 0.3$, stops short of the top and is under-approved, with $\min_{p > p_{\min}} \tau = 0.9862$ (Table~\ref{tab:numerics}, row 2).

\subsection{The Critical Slope}\label{sec:critical}

The ramp and the step differ in regularity but agree in what they must do at one point: both saturate at $r_0$, and both must climb to saturation fast enough. The next theorem shows that this is the only requirement. Write $\Lambda$ for a \emph{slope budget}, $\mathcal Q_\Lambda \coloneqq \{q : [0,1] \to [0,1] \text{ non-decreasing},\ \mathrm{Lip}(q) \le \Lambda\}$ for the rules that respect it, and $\delta_\Lambda(G) \coloneqq W^* - \sup_{q \in \mathcal Q_\Lambda} W(q)$ for the welfare gap it imposes. The \emph{critical slope} is the steepest chord of the ramp on the way to the reserve,
\begin{equation}\label{eq:critical-slope-def}
\Lambda_c(G) \;\coloneqq\; \sup_{a \in [p_{\min},\,r_0)} \frac{1 - \mathrm{ramp}(a)}{r_0 - a} \;=\; \frac{\beta}{\gamma}\, \sup_{a \in [p_{\min},r_0)} \frac{C(r_0) - C(a)}{r_0 - a} .
\end{equation}
When the ramp slope $s$ is non-decreasing on $[p_{\min},r_0]$, which holds for the Brier, power and log scores (Example~\ref{ex:log-score}), the supremum is the limit of the chords at the reserve and the critical slope has the closed form $\Lambda_c(G) = s(r_0) = (\beta/\gamma)\,G''(r_0)\,(r_0 - p_{\min})$, the ramp's slope at the reserve; \suppsec{S.9} gives the exact criterion for the closed form and an engineered generator at which the chord supremum exceeds the ramp's slope at the reserve by a factor above five (Table~\ref{tab:numerics}, row 3). Lemma~\ref{lem:structure} pins both ends of every correct rule's climb, and the next theorem turns that into the exact slope requirement.

\begin{theorem}[First-best screening needs exactly the critical slope]\label{thm:critical-slope}
For rules $q : [0,1] \to [0,1]$ and slope budgets $\Lambda > 0$:
\begin{enumerate}
\item[(i)] \emph{(The chord bound.)} Every continuous rule $q$ that attains first-best screening satisfies, for every $a \in [p_{\min},r_0)$,
\[
\frac{q(r_0)-q(a)}{r_0-a} \;\ge\; \frac{1-\mathrm{ramp}(a)}{r_0-a} ;
\]
hence $\mathrm{Lip}(q) \ge \Lambda_c(G) > 0$, with $\Lambda_c$ as in~\eqref{eq:critical-slope-def}.
\item[(ii)] \emph{(The threshold, both directions.)} Within the slope budget $\Lambda$, first-best screening is attainable if and only if $\Lambda \ge \Lambda_c(G)$, and the truncated line $q_\Lambda(r) \coloneqq \operatorname{med}\{0,\ 1-\Lambda(r_0-r),\ 1\}$ attains it, with a unique best response for every $p \ne p_{\min}$. Consequently $\delta_\Lambda(G) = 0$ if and only if $\Lambda \ge \Lambda_c(G)$, and the threshold slope $\Lambda_c$ is itself admissible, so the infimum over admissible slopes is a minimum.
\end{enumerate}
\end{theorem}

Part~(i) uses no differentiability of $q$. Part~(ii) holds for the unrestricted Lipschitz class as well as for the monotone one, because necessity is proved over the unrestricted class and the construction is monotone, and its welfare equivalence rests on (S) alone, without the monotone ramp slope that Proposition~\ref{prop:below-threshold} assumes. Proof in Appendix~\ref{app:critical-slope}.

\begin{corollary}[No $C^1$ rule attains first best]\label{cor:c1}
No continuously differentiable rule $q : [0,1] \to [0,1]$ attains first-best screening; continuity on $[0,1]$ together with differentiability at the single point $r_0$ already excludes it. The infimum in Theorem~\ref{thm:kink}(ii) is therefore not attained.
\end{corollary}

Strict feasibility cannot be dropped: at the boundary prize, where $r_0 = 1$, the Brier ramp is $C^1$ on $[0,1]$ and attains first best by Theorem~\ref{thm:kink}(i). Proof in Appendix~\ref{app:critical-slope}.

On the monotone class the cap and the saturation of Lemma~\ref{lem:structure} are exactly the hypotheses of Lemma~\ref{lem:envelope}, which closes the attainment results into one statement.

\begin{corollary}[The first-best rules]\label{cor:first-best}
Let $q : [0,1] \to [0,1]$ be non-decreasing and admit a global best response at every type (automatic if $q$ is upper semicontinuous). Then $q$ attains first-best screening if and only if $q \le \mathrm{ramp}$ on $[0,1]$ and $q(r_0) = 1$. The step rule of Theorem~\ref{thm:reserve}, the Bregman ramp of Theorem~\ref{thm:kink} and the truncated line $q_\Lambda$ of Theorem~\ref{thm:critical-slope}(ii), for any $\Lambda \ge \Lambda_c(G)$, are three named members of that family.
\end{corollary}

The characterization holds under (A1) and weak feasibility; only the truncated-line clause uses (S). Proof in \suppsec{S.8}.

\emph{Intuition.} The binding constraint is how fast the rule may change, not how smoothly. The rule must climb from a point underneath the ramp to the value $1$ at a fixed report within a fixed horizontal distance, and its path in between is irrelevant, because only the two endpoints matter and both are pinned by Lemma~\ref{lem:structure}. The obstruction is therefore a family of chords rather than a tangent, and no cleverer non-linear path evades it. A $C^1$ rule must flatten where it saturates, and a flat top invites under-reporting, because the calibration saving from retreating is first order while the approval loss is not (Corollary~\ref{cor:c1}).

\paragraph{Score curvature and the slope premium.}
Under a monotone ramp slope the critical slope factorizes exactly. Since $C(r_0) = \int_{p_{\min}}^{r_0} \kappa(t)(t-p_{\min})\,dt = \langle\kappa\rangle_\chi\, m^2/2$ for $\kappa \coloneqq G''$ and the distance-weighted probability measure $d\chi(t) = 2(t-p_{\min})\,dt/m^2$ on $[p_{\min},r_0]$, substituting into $\Lambda_c = \kappa(r_0)m/C(r_0)$ gives
\begin{equation}\label{eq:tilt-identity}
\Lambda_c(G) \;=\; \frac{2}{m}\cdot\frac{\kappa(r_0)}{\langle\kappa\rangle_\chi} .
\end{equation}
Hold the induced miscalibration $m$ fixed, so that the scores compared produce the same reported distortion; then the whole dependence on the score is the dimensionless ratio of its curvature \emph{at} the reserve to the distance-weighted average of its curvature \emph{below} it. Flat curvature, the Brier class, gives ratio $1$ and is the exact minimizer of the critical slope among all scores whose curvature is largest at the reserve, a class containing every power score with $k \ge 2$, their non-negative mixtures, and the log score exactly when $p_{\min} + r_0 \ge 1$; at lower reserves the log score's curvature peaks at $p_{\min}$ instead, it leaves the class, and its slope premium against Brier can take either sign. Scores that sharpen toward certainty need strictly more slope, and the premium is this proportional gap: at $p_{\min} = 0.4$, $r_0 = 0.7$, the log score costs $11.6\%$ more slope than Brier and $G = p^4$ costs $34.2\%$ (Table~\ref{tab:numerics}, row 5). Gains in the other direction are capped by $\Lambda_c \ge 1/m$, the chord from $p_{\min}$: choosing the score badly can make oversight arbitrarily harder, whereas choosing it well buys at most a factor of two.

\subsection{Below the Critical Slope}

Above the critical slope the design problem has no cost; below it, the cost is third order in the shortfall. Write $f_{\min} \coloneqq \operatorname{ess\,inf} f > 0$ and $\zeta \coloneqq 1 - \Lambda/\Lambda_c(G) \in (0,1)$ for the relative shortfall of a budget $\Lambda < \Lambda_c(G)$.

\begin{proposition}[Two-sided welfare bounds below $\Lambda_c$]\label{prop:below-threshold}
Assume (A3) and that the ramp slope $s$ is non-decreasing on $[p_{\min},1)$ with $s'(r_0) > 0$, which holds for every power score and for the log score, and let all statements be understood asymptotically as $\Lambda \uparrow \Lambda_c(G)$, equivalently as $\zeta \downarrow 0$. Then $\delta_\Lambda(G) \asymp \zeta^3$, in both directions. \emph{Upper bound:} writing $\hat q_\Lambda$ for the largest $\Lambda$-Lipschitz function below the ramp, $\delta_\Lambda(G) \le W^* - W(\hat q_\Lambda)$, and that gap is a constant times $\zeta^3(1 + o(1))$. \emph{Lower bound:} there is $\zeta_0 > 0$ such that for $\zeta < \zeta_0$ every $\Lambda$-Lipschitz $q$, monotonicity not assumed, has $W^* - W(q)$ at least a constant times $\zeta^3(1 - O(\zeta))$. Appendix~\ref{app:below-threshold} derives both constants; they differ by the factor $f(p_{\min})/(f_{\min}\Gamma_1^2)$, where $\Gamma_1 \in (0,1)$ is the generator-dependent constant of the lower bound, equal to $\tfrac12$ under Brier, so the ratio is $4$ for Brier under uniform $F$.
\end{proposition}

Both bounds use (A3), through the continuity of $s'$ near the reserve. The exact minimizer over $\mathcal Q_\Lambda$, hence the exact constant, is not identified, and the bounds are asymptotic in $\zeta$ rather than uniform over $\Lambda \in (0,\Lambda_c)$, since far below the critical slope the affected band is no longer a neighborhood of $p_{\min}$, leaving only $0 < \delta_\Lambda \le W^* - u_d$. Without the monotone ramp slope the critical slope need not equal the ramp's slope at the reserve (\suppsec{S.9}), the hypothesis fails and the rate is open.

\emph{Intuition.} A slope cap fails first for the types just above $p_{\min}$, which have the least to gain from climbing and stop short first when the rule stops rewarding the climb fast enough. Three powers of the shortfall accumulate: the affected band is $O(\zeta)$ wide, each type in it loses $O(\zeta)$ of approval probability, and the welfare weight $|p - p_{\min}|$ it carries is itself $O(\zeta)$ there. The exponent is therefore three: a rule a fraction of a percent below critical is almost indistinguishable from an exact one, and the slope budget matters only when it is tight.

\subsection{Distribution-Free Design}

The reserve and the critical slope, which fix where oversight happens, are functionals of the score, the threshold and the prize, and the type distribution enters only the price of falling short.

\begin{proposition}[The reserve and the critical slope are distribution-free]\label{prop:f-free}
$r_0(G)$ and $\Lambda_c(G)$ are functionals of $(G,p_{\min},\gamma/\beta)$ alone. They do not depend on $F$, and they do not depend on the payoffs $(u_s,u_f,u_d)$ except through the single number $p_{\min} = (u_d-u_f)/(u_s-u_f)$. Consequently the set of slope budgets that admit first-best screening is $[\Lambda_c(G),\infty)$, and it is distribution-free: the same approval rule screens correctly under every full-support type distribution. In the saturated regime $\gamma/\beta > D_G(p_{\min},1)$ that set is empty for every $F$, and the second-best intensity $\lambda_{\mathrm{sat}}$ of Proposition~\ref{prop:lambda-sat} is again a functional of $(G,p_{\min},\gamma/\beta)$ alone, so distribution-freeness holds across the whole parameter space.
\end{proposition}

The reserve is distribution-free under weak feasibility alone. The proof is immediate: every object in the statement is built from $C$, $r_0$ and $\gamma/\beta$, into which $F$ does not enter and the payoffs enter only through $p_{\min}$; the per-type screening statements of Theorems~\ref{thm:reserve}(i), \ref{thm:kink}(i) and~\ref{thm:critical-slope}(ii) hold under every $F$; and the necessity direction, Lemma~\ref{lem:structure} with Theorem~\ref{thm:critical-slope}(i), or Corollary~\ref{cor:first-best} at the boundary $r_0 = 1$, uses $F$ only through full support. The content nonetheless reverses the classical screening intuition.

\paragraph{The Myerson contrast.}
In \citet{myerson1981optimal} the reserve is a distributional object, solving $\psi(v) = 0$ for the virtual valuation $\psi(v) = v - (1-F(v))/f(v)$, so it moves with the hazard rate and a seller who misestimates $F$ sets the wrong reserve. Here it is behavioral: $r_0$ is the report at which the marginal type's calibration cost of pretending exhausts the approval prize. Myerson's reserve trades rent extraction against exclusion, a statement about the mass of types on either side of it; ours trades a known curvature against a known prize, a statement about one type's indifference, and a \emph{measured} cost of misreporting removes the distribution from the threshold. The threshold structure, the Stackelberg timing and the incentive-constrained optimization carry over exactly; the virtual-type construction has no distributional content here, and first best is attainable, whereas Myerson's optimum entails allocative inefficiency. The same mechanism separates the result from the optimal regulation of \citet{laffont1993theory}: there information rents force a wedge between first and second best, here the wedge is zero because the calibration cost of pretending is a known function of the score, and the role of the hazard rate passes to the score's curvature through $\Lambda_c$.

The distribution enters in exactly two places, both below the threshold or off it. First, the welfare weight in~\eqref{eq:welfare-gap-identity}: $F$ determines what a screening error costs, so the constants in Proposition~\ref{prop:below-threshold} involve $F$ only through $f(p_{\min})$ and $f_{\min}$, and the leading constant depends on the density only \emph{at the marginal type $p_{\min}$}, since the affected band shrinks to $p_{\min}$ as $\zeta \downarrow 0$. Second, the comparison between scores below $\Lambda_c$, which the next proposition settles at a fixed induced miscalibration. Write $\varkappa(G) \coloneqq \kappa(r_0)/\langle\kappa\rangle_\chi$ for the curvature ratio of~\eqref{eq:tilt-identity} and $\xi(G) \coloneqq m\,\kappa'(r_0)/\kappa(r_0)$ for the \emph{sharpening} of the score at the reserve.

\begin{proposition}[Score selection below the critical slope]\label{prop:score-selection}
Assume the hypotheses of Proposition~\ref{prop:below-threshold} and compare scores at a fixed induced miscalibration $m$, so that the prize $\gamma/\beta = D_G(p_{\min},r_0)$ varies with the score. Then the two constants of Proposition~\ref{prop:below-threshold} depend on the type distribution only through $f(p_{\min})$ and $f_{\min}$, which are common to every score at the same $p_{\min}$, so the ranking of scores by either constant at a given relative shortfall $\zeta$ is distribution-free; and relative to the Brier score at the same $m$, a score's upper constant is $\varkappa/(1+\xi)$ times Brier's while its critical slope is $\varkappa$ times Brier's.
\end{proposition}

The closed forms of both constants, and the inequality $\varkappa \le 1 + \xi$ for scores whose curvature is convex up to the reserve with $\xi \le 2$, are Proposition~\ref{prop:score-selection-full} in Appendix~\ref{app:score-selection-full}. The design question thus splits in two. Above $\Lambda_c(G)$ every admissible score attains first best and the choice is made on other grounds, such as elicitation robustness or interpretability. Below it, a score whose curvature is convex on the way to the reserve and largest there, with sharpening $0 \le \xi \le 2$, lowers the loss bound of Proposition~\ref{prop:below-threshold} per unit of relative shortfall, by the factor $\varkappa/(1+\xi) \le 1$, and raises the slope it needs, by the factor $\varkappa \ge 1$: at $p_{\min} = 0.4$, $r_0 = 0.7$ the log score and $G = p^4$ carry $0.710$ and $0.723$ of Brier's upper constant while needing $11.6\%$ and $34.2\%$ more slope (Table~\ref{tab:numerics}, row 5). The comparison is at equal relative shortfall; at a common budget $\Lambda$ the two shortfalls differ, Brier's threshold being the lowest among scores whose curvature is largest at the reserve, so the score that is nearest its own threshold can win outright.

\begin{example}[The log score]\label{ex:log-score}
Take the negative-entropy generator $G(p) = p\ln p + (1-p)\ln(1-p)$, whose score is the log score and whose divergence is $D_G(p,r) = \mathrm{KL}(p\,\|\,r)$, with $G''(p) = 1/(p(1-p))$ unbounded at both endpoints. Feasibility is automatic and strict, since $D_G(p_{\min},1) = +\infty$ for every $p_{\min}$: under the log score there is no infeasible regime and no saturated regime at all. The reserve is interior and is the edge of the Kullback--Leibler ball of radius $\gamma/\beta$ around the marginal type, the unique $r_0 > p_{\min}$ with $\mathrm{KL}(p_{\min}\|r_0) = \gamma/\beta$. That equation is transcendental, but the small-prize expansion $r_0 = p_{\min} + \sqrt{2p_{\min}(1-p_{\min})\,\gamma/\beta} + O(\gamma/\beta)$ exhibits the Brier form with the Fisher-information scale $1/(p_{\min}(1-p_{\min}))$ in place of the constant $2$. The ramp slope is $s(r) = (\beta/\gamma)(r - p_{\min})/(r(1-r))$, and since $\frac{d}{dr}\frac{r-p_{\min}}{r(1-r)} = \frac{r^2-2p_{\min}r+p_{\min}}{(r(1-r))^2}$ has numerator of discriminant $4p_{\min}(p_{\min}-1) < 0$, hence strictly positive, $s$ is strictly increasing at every $(p_{\min},\gamma/\beta)$ and the critical slope is the closed form $\Lambda_c = (\beta/\gamma)(r_0-p_{\min})/(r_0(1-r_0))$, reached with no bounded-curvature hypothesis anywhere; at $p_{\min} = 0.6$, $\gamma/\beta = 0.04$ this gives $r_0 = 0.730683$ and $\Lambda_c = 16.602$ (Table~\ref{tab:numerics}, row 4).
\end{example}

\subsection{The Saturated Regime}\label{sec:saturated}

When the approval prize exceeds the report space's whole calibration budget, $\gamma/\beta > D_G(p_{\min},1)$, no reserve exists and no deterministic rule screens; this regime lies outside (S), and the result below assumes (A1) in its place. Partial approval, which lies inside the stated design space, still has value in this regime.

\citet{loven2026behavioral} characterize the regime: they exhibit the boundary pooling at $r_0 = 1$, observe that randomized approval strictly improves on the deterministic second-best, and pose the closed form for arbitrary $F$ as open. Proposition~\ref{prop:lambda-sat} answers the exact-screening part of that question: it identifies the point of the second-best frontier at which the induced cutoff is exactly $p_{\min}$. The unconstrained optimum for arbitrary $F$ remains open.

\begin{proposition}[Exact screening at reduced intensity]\label{prop:lambda-sat}
Assume (A1) and the saturated regime, $\gamma/\beta > D_G(p_{\min},1)$. Put $\lambda_{\mathrm{sat}} \coloneqq \beta D_G(p_{\min},1)/\gamma \in (0,1)$ and $q_{\mathrm{sat}} \coloneqq \lambda_{\mathrm{sat}}\,\mathbf 1\{r \ge 1\}$. Then $q_{\mathrm{sat}}$ induces $\tau(p) = \lambda_{\mathrm{sat}}\,\mathbf 1\{p \ge p_{\min}\}$ for every $p \ne p_{\min}$, so the screening boundary lands exactly at the first-best cutoff, and
\[
W(q_{\mathrm{sat}}) \;=\; u_d + \lambda_{\mathrm{sat}}\int_{p_{\min}}^1 \Pi(p) f(p)\, dp \;>\; u_d .
\]
This is optimal within the class of approval rules that reject almost every type below $p_{\min}$.
\end{proposition}

Probabilistic approval substitutes for an infeasible calibration budget: scaling the stake down to $\lambda_{\mathrm{sat}}\gamma$ scales the deterrence requirement down with it, until the cost of reporting $r = 1$ is again exactly what the marginal type will pay. No later result depends on this regime.

\section{Marketplace Operation}\label{sec:market}

This section applies the elicitation-theoretic frame of Section~\ref{sec:behavioral} to a second and economically independent setting, a marketplace operator who executes an allocation-payment mechanism on behalf of bidders. The explicit model makes the abstract objects of Section~\ref{sec:model} concrete against a classical mechanism-design construction, and places them where the credibility question has already been posed and answered by other means, so the frame can be checked against known results rather than only against itself. The market credibility theory proper, including the trilemma between revenue optimality, agent incentive compatibility and operator credibility on polymatroidal feasible regions, is developed in a companion paper \citep{loven2026trilemma}, which we cite for everything market-specific.

The mapping to the credibility game of Definition~\ref{def:credibility-game} is: $\theta = \mathbf{b}$ (the true bid profile), $r = \hat{\mathbf{b}}$ (effective bids), $g(\mathbf{b}) = \mathbf{b}$ (faithful execution), $\Sbar = -\delta_{\mathrm{rep}}\|\hat{\mathbf{b}} - \mathbf{b}\|^2$ (reputational compliance), and $h = R$, where $R(\hat{\mathbf{b}}) = \sum_i p_i^*(\hat{\mathbf{b}})$ is total DSIC revenue under the \citet{archer2001truthful} payment identity:
\begin{equation}\label{eq:AT-payment}
p_i^*(\hat{\mathbf{b}}) = \hat{b}_i x_i^*(\hat{\mathbf{b}})
- \int_0^{\hat{b}_i} x_i^*(z, \hat{\mathbf{b}}_{-i})\, dz,
\end{equation}
with $x^*$ the Edmonds greedy allocation, exact on a polymatroid \citep{conforti1984submodular}. The combined objective~\eqref{eq:combined-objective} is then $V(\hat{\mathbf{b}}; \mathbf{b}, \gamma) = -\delta_{\mathrm{rep}}\|\hat{\mathbf{b}} - \mathbf{b}\|^2 + \gamma R(\hat{\mathbf{b}})$.

\begin{definition}[Marketplace Game]\label{def:marketplace-game}
A \emph{marketplace game} is $\calG_M = (\calN, \Theta_i, F_i, x^*, p^*, \nu, \calI_M)$: $n$ agents with types $v_i \in [0, \bar{v}]$ drawn independently from distributions $F_i$ with continuous densities $f_i > 0$; the Edmonds greedy allocation $x^*(\mathbf{b}) \in P(\nu) = \{x \in [0,1]^n : \sum_{i \in T} x_i \leq \nu(T)\;\forall\, T \subseteq \calN\}$ for monotone submodular $\nu$ with $\nu(\emptyset) = 0$ and $\nu(\{i\}) \le 1$ \citep[cf.][]{conforti1984submodular}, monotone in each bid so that the payments~\eqref{eq:AT-payment} are DSIC; non-modularity gap $\kappa_{ij} = \nu(\{i\}) + \nu(\{j\}) - \nu(\{i,j\}) \geq 0$; and sealed-bid information structure $\calI_M$ where agent $i$ observes only $(x_i(\hat{\mathbf{b}}), p_i(\hat{\mathbf{b}}))$. Ties in the greedy ordering are broken by agent index, so $x^*$ and the payments~\eqref{eq:AT-payment} are single-valued functions of $\hat{\mathbf{b}}$. The DSIC equilibrium concept requires $b_i = v_i$ to be dominant given faithful execution.
\end{definition}

A notational caution for readers holding both papers: our $\gamma$ is the weight on the perturbation payoff, whereas \citet{loven2026trilemma} write $\gamma_{ij \mid S}$ for a non-modularity gap conditional on the set $S$ of agents already processed, of which our $\kappa_{ij}$ is the case $S = \emptyset$.

The operator's profitable deviation is an inflation of a lower-priority bid: raising $\hat{b}_j$ by $\Delta b$ while the greedy ordering is preserved leaves a higher-priority agent $i$'s allocation unchanged but contracts the integral in~\eqref{eq:AT-payment}, raising $i$'s payment by $\Delta b\,\kappa_{ij}$ on the welfare branch. For $n \geq 3$ the governing quantity is the conditional gap rather than the unconditional one, and the exact increment, its contested window, its revenue-optimal branch and the welfare accounting of the resulting transfer are established in \citet{loven2026trilemma}; we use only that the increment is strictly positive on a set of profiles of positive measure. For $n = 2$ the section is self-contained: when $\nu(\{i\}) \leq 1$, Edmonds greedy and~\eqref{eq:AT-payment} give $R(\hat{\mathbf{b}}) = \kappa_{12}\min(\hat{b}_1, \hat{b}_2)$ exactly, so the increment is $\kappa_{12}$ on each ordering region.

Fix a greedy ordering $o$, let $\calR_o \subseteq [0,\bar{v}]^n$ be the closed set of profiles it ranks, on which $R$ is affine and hence $C^\infty$, and write $\calG_M^o$ for the credibility game of Definition~\ref{def:credibility-game} with reports and types restricted to $\calR_o$. The restriction makes Assumption~\ref{assum:regularity} checkable here: $\calR_o$ is convex and compact with non-empty interior, so~(iii) holds; $g(\mathbf{b}) = \mathbf{b}$ is interior for every profile off the $\mu$-null boundary $\partial\calR_o$, as the footnote to Definition~\ref{def:credibility-game} allows, so~(i) holds where $\Cbind$ is taken; and~(ii) is the hypothesis Corollary~\ref{cor:market} places on $\Sbar$, met by the quadratic compliance score with $-\nabla^2_{\hat{\mathbf{b}}}\Sbar = 2\delta_{\mathrm{rep}} I$. NT1 holds because faithful execution does not maximize revenue when some pair has a strictly positive gap. The gradient clause of NT2, the clause Lemma~\ref{lem:perturbation}(ii) uses, holds on every region on which some bid carries a strictly positive one-sided derivative of $R$, as the second-ranked bid does whenever the top-ranked pair has $\kappa_{ij} > 0$, its derivative there being $\kappa_{ij}$; the non-affinity clause does not, since $R$ is affine on each region and non-affine only across the $\mu$-null set of profiles at which the ordering changes, so Corollary~\ref{cor:market} invokes the lemma rather than Definition~\ref{def:nontrivial} in full. That also separates this instance from the benign case of Section~\ref{sec:model}: inside one region $R$ is affine and the induced shift is common to every type there, but the shift differs across regions, so the type-dependence comes from which region a type sits in rather than from curvature within one, and a receiver who cannot place $\mathbf{b}$ in a region cannot undo it. NT3 holds because under $\calI_M$ agent $i$'s signal $(x_i^*, p_i^*)$ has the same distribution whether the operator inflated $b_j$ or agent $j$ bid $\hat{b}_j$, provided the inflated profile still lies in $\calR_o$, which the separation condition
\begin{equation}\label{eq:market-separation}
\min_{i \neq j} |b_i - b_j| \;>\; \frac{\gamma}{2\delta_{\mathrm{rep}}}\,\|\nabla R\|_\infty
\end{equation}
secures, since under the quadratic score the optimal deviation is $\hat{\mathbf{b}} = \mathbf{b} + (\gamma/2\delta_{\mathrm{rep}})\nabla R$ with $\nabla R \ge 0$ constant on the region, so no two bids cross; $\|\nabla R\|_\infty$ is $\kappa_{12}$ for $n = 2$ and for $n \geq 3$ aggregates the conditional gaps of \citet{loven2026trilemma}, and the separated set carries positive measure whenever $(n-1)\gamma \|\nabla R\|_\infty < 2\delta_{\mathrm{rep}}\bar{v}$.

\begin{corollary}[Market credibility impossibility]\label{cor:market}
Let $\calG_M$ be a marketplace game under sealed-bid execution whose polymatroidal feasible region is non-modular, $\kappa_{ij} > 0$ for some pair, let $\calG_M^o$ be its restriction to a greedy-ordering region, and let $\Sbar$ be strictly proper with $g(\mathbf{b}) = \mathbf{b}$ and satisfy Assumption~\ref{assum:regularity}(ii) on $\calR_o$. Then NT1 and the gradient clause of NT2 hold on $\calG_M^o$ as verified above, so Lemma~\ref{lem:perturbation}(ii) applies: no operator strategy is simultaneously faithful~(T) and rational~(R) under the combined objective $V$. Under the quadratic compliance score, on the separated set~\eqref{eq:market-separation} the optimal deviation stays inside the region, and NT3, delivered there by the sealed-bid structure, leaves it unpunishable and licenses the word impossibility in the sense reserved in Section~\ref{sec:fenchel}.
\end{corollary}

\begin{remark}[Relation to Akbarpour--Li and to the polymatroid trilemma]\label{rem:al-relationship}
\citet{akbarpour2020credible} establish their impossibility through extensive-form sequential rationality, and \citet{loven2026trilemma} extend that frame to polymatroidal feasible regions, proving under competitive priors a three-legged trilemma with the conditional gap as its primitive. Corollary~\ref{cor:market} is neither result and subsumes neither. It is elicitation-theoretic: the operator is scored by a strictly proper compliance rule, and undetectability is Bayesian indistinguishability of the agents' signal distributions (NT3), a strictly stronger requirement than the support-style safe-deviation property under which the extensive-form results are proved. Its conclusion is correspondingly weaker, two legs rather than three, on the welfare branch and without a competitive-prior hypothesis. The two share one economic force, that conflicting objectives destroy credibility, but neither implies the other.
\end{remark}

\paragraph{Compliance score: reduced-form status and form dependence.}
The quadratic compliance score $\Sbar = -\delta_{\mathrm{rep}}\|\hat{\mathbf{b}} - \mathbf{b}\|^2$ is a reduced-form assumption, not derived from Savage--McCarthy foundations. Properness can arise endogenously from reputation dynamics: the career-concerns logic of \citet{holmstrom1999managerial} and \citet{mailath2001pays} makes a patient platform's long-run objective loss-minimizing under Bayesian updating on outcomes, which by the de~Finetti--Savage characterization \citep{definetti1937prevision, savage1971elicitation} corresponds to maximizing a proper scoring rule. Formalizing that requires a state space, an outcome mapping and a loss function for the reputation game, which we do not supply, and we claim nothing beyond the reduced form. The form dependence is then explicit. Corollary~\ref{cor:market} does not use the quadratic shape; it uses Assumption~\ref{assum:regularity}(ii), which strict properness on its own does not deliver and which Lemma~\ref{lem:perturbation} needs. The quadratic form buys exactly one thing, a closed form for the size of the deviation through the constant Hessian $-2\delta_{\mathrm{rep}} I$ in~\eqref{eq:perturbation-formula}, and it buys it on the separated set~\eqref{eq:market-separation}: off that set the shifted report crosses into a neighboring greedy region, where $R$ switches branch, and the formula no longer locates the optimum. Smoothness, by contrast, is not cosmetic: a kinked score such as $-\delta_{\mathrm{rep}}\|\hat{\mathbf{b}} - \mathbf{b}\|_1$, strictly proper in the sense of Definition~\ref{def:strict-properness} though it admits no Savage generator (Remark~\ref{rem:savage-characterization}), pins faithful execution as the unique maximizer whenever $\gamma \sup \|\nabla R\|_\infty < \delta_{\mathrm{rep}}$, and the impossibility fails.

\paragraph{Resolution.} An ascending auction with public broadcast makes any operator deviation detectable, since agents observe the entire price path, and so eliminates NT3 at the cost of bid privacy; \citet{loven2026trilemma} execute that program for this market, proving the ascending clinching auction credible and incentive compatible under an explicit broadcast-channel model. The instance thus contributes a reading rather than a market result: operator credibility here is an elicitation problem, and its failure is the same argmax shift that Section~\ref{sec:behavioral} showed an oversight designer cannot escape by design.

\section{Discussion}\label{sec:discussion}

\subsection{The Endogeneity}

Two familiar phenomena sit close to the endogeneity, the principal's own optimum being the perturbation, and each differs in a specific way.

\emph{Goodhart's Law} \citep{goodhart1984problems} and \emph{the Lucas critique} \citep{lucas1976econometric}. Both warn that behavior shifts once the measure or the policy evaluating it is known; \eqref{eq:perturbation-formula} gives that shift in closed form, a quantitative instance of causal Goodhart \citep{manheim2018categorizing}, and here the regime is the design variable rather than a hazard to extrapolation.

The critical slope also reconciles two accounts of the agent that would otherwise read as competing modeling assumptions. Above $\Lambda_c(G)$ the best response sits at a corner and the induced screening is the first-best indicator exactly; below it the best response is interior and the screening is a smooth profile of partial approval probabilities (Theorem~\ref{thm:critical-slope}, Proposition~\ref{prop:below-threshold}). Which of the two an analyst observes is set by the slope budget the designer works under, not by an assumption about the agent's rationality.

\subsection{Screening-Sufficient and Calibration-Essential Applications}

Inflated reports are compatible with first-best decisions, so a principal need not always care about the inflation. The line is decidable at the modeling stage: whether welfare depends on the report only through $q(r)$, as in~\eqref{eq:principal-utility}, or on $r$ itself, as when another party consumes it.

\emph{Screening-sufficient applications} are of the first kind, and in them the report is an intermediate object no other party consumes. An approval gate for autonomous execution is one, where an agent requests permission to act without human review and the overseer grants or withholds it; admission to a marketplace is another, where an operator admits or excludes a participant on a self-reported capability score that is discarded once admission is settled (Section~\ref{sec:market}). In both, the threshold absorbs the predictable inflation, and, where the reserve is feasible, a principal who insists on truthful reports pays for them without buying anything (Theorem~\ref{thm:welfare-equivalence}, Appendix~\ref{app:benchmark}).

\emph{Calibration-essential applications} are those in which the report is itself an argument of welfare, by one of three routes. The number may be consumed by a downstream decision maker who acts on it rather than on the approval, as with a triage score read by a clinician: the overseer is then one consumer among several, and the others receive a shifted number with no threshold to reinterpret it. Reports may be aggregated, so that an inflation common to a pool of forecasters survives averaging instead of cancelling. Or the report may be priced into a premium or a capital charge, so that a shift in the report is a shift in money. In all three the report must be corrected rather than absorbed.

External regulation earns its cost in exactly the calibration-essential case, and in the screening-sufficient case only when oversight is not doing its job. The gain is at most the avoided harm from mis-approving types below $p_{\min}$, so regulation improves welfare only when $\Pr(p < p_{\min}) \cdot \E[|\Pi(p)| \mid p < p_{\min}]$ exceeds the cost of regulating, and in the marketplace instance when the welfare loss from bid inflation under sealed-bid execution exceeds the cost of mandating ascending formats.

\subsection{The Price of the Restricted Architecture}

\paragraph{Relation to the revelation principle.}
Reports are inflated in equilibrium throughout this paper. The revelation principle says that every equilibrium outcome of every mechanism is reproduced by some direct mechanism under which truthful reporting is optimal, and it holds in this setting, with Theorem~\ref{thm:canonical-form} giving the canonical form of the direct mechanism it lands in. The relabelling changes the transfer and not only the message: the direct mechanism reproducing the outcome of an approval rule pays a proper score for the report and charges exactly $\gamma$ on the approval event, and that charge is the instrument the architecture of Definition~\ref{def:oversight-game} does not have. The architecture therefore restricts the instrument set rather than the message space, and it is not without loss of generality: its incentive-compatible members are exactly the rules constant in the report (Theorem~\ref{thm:welfare-equivalence}(iii)), which screen nobody, so a principal that screens at all within it induces misreporting. The loss falls on truthfulness and not on welfare, since wherever the reserve is feasible the architecture attains what the unrestricted class attains.

Appendix~\ref{app:benchmark} prices the two deletions in coordinates and carries the algebra: freezing the generator is innocuous, and deleting the approval charge is the whole restriction. The extra instrument buys an honest message and no welfare, at a charge equal to the agent's own value of autonomy, and an agent with no budget and a reward channel fixed upstream by its training cannot be charged in that currency.

\subsection{Transfer to Deployed Systems}

The abstracted structure is recognizable in deployed systems: a model states a confidence, an overseer grants or withholds permission to act on it without human review, and the model's training rewards both accuracy and being permitted to act. Where a learned reward model scores both, $\gamma/\beta$ is the weight on being permitted to act relative to the weight on being accurate, and it is a design variable: lowering it, by weighting accuracy more heavily in high-stakes deployments or by making permission less rewarding, moves the reserve $r_0(G)$ toward $p_{\min}$ and shrinks the induced inflation without changing which requests are granted. The second lever is where the rule may be steep, and a cliff is not required: climbing at rate $\Lambda_c(G)$ suffices, and climbing more slowly loses welfare only at third order in the shortfall (Theorem~\ref{thm:critical-slope} and Proposition~\ref{prop:below-threshold}).

Three features do not transfer without further work. The overseer is assumed to know $\gamma$ and $\beta$, properties of the agent's training rather than of its own choice, so the reserve is only as accurate as the estimate of that ratio. The interaction is one-shot. The type is scalar with a binary outcome and a single decision, whereas a deployed agent reports on several dimensions at once; Appendix~\ref{sec:multi-dim} records what carries to $d$ dimensions, where the reserve is a level set of the divergence rather than a point.

\subsection{Open Questions}

Five directions are open above: describing the slope-premium-minimizing class in terms of the score, not its generator's curvature; the exact constant in the cubic law and its exponent in the case where the critical slope exceeds the ramp's slope at the reserve, the value of an obstacle problem with an endogenous obstacle, because a rule may rise above the ramp at a price paid by types below $p_{\min}$ (Proposition~\ref{prop:below-threshold}); whether $q_{\mathrm{sat}}$ beats every rule approving types below $p_{\min}$ (Proposition~\ref{prop:lambda-sat}); the $\R^d$ analogue of the critical slope (Appendix~\ref{sec:multi-dim}); and which non-proper losses, paired with which approval rules, dominate the proper-score architecture, since a principal willing to abandon properness gives up the elicitation anchor that defines the design problem here.

The model is static, and the sixth direction is dynamic monitoring. Repeated interaction and reputation dynamics \citep{fudenberg1989reputation} may partially restore truthfulness: an overseer able to credibly threaten revocation of autonomy on detecting miscalibration lowers the agent's effective value of approval by an amount increasing in the discount factor. The $\Theta(1/\Delta^2)$ detection horizon of Section~\ref{sec:game} supplies the detection side of that threat; the equilibrium of the repeated game, and whether the critical slope falls as detection improves, are not treated here.

\section{Conclusion}\label{sec:conclusion}

A principal who screens an agent on a scored report cannot use an affine approval rule, because an affine rule cannot deliver first-best screening. The welfare-optimal rule is therefore non-affine whenever the approval prize is feasible for the marginal type (Theorem~\ref{thm:optimal-nonaffine}), and a non-affine, report-dependent payoff displaces the agent's optimal report away from the truth. The design that maximizes the principal's welfare is the design that destroys the calibration of the reports it reads.

Because the distortion is predictable, the rule can accommodate it. Setting the bar at the reserve report $r_0(G)$, where the marginal type's calibration cost of pretending exhausts the approval prize, screens every type correctly at no welfare cost even though nothing undoes the inflation (Theorem~\ref{thm:reserve}); an unrestricted mechanism buys back the truthful message and nothing else, at a charge equal to the agent's own value of autonomy (Theorem~\ref{thm:welfare-equivalence}). Steepness rather than smoothness binds the implementation. A Lipschitz rule with a single kink induces the same outcome for every type but one, so no regulator needs to implement a cliff, while under strict feasibility no continuously differentiable rule attains the first best exactly (Theorem~\ref{thm:kink}, Corollary~\ref{cor:c1}); first-best screening is attainable within a slope budget $\Lambda$ if and only if $\Lambda \geq \Lambda_c(G)$ (Theorem~\ref{thm:critical-slope}), and below that budget the welfare cost is cubic in the shortfall (Proposition~\ref{prop:below-threshold}). The reserve and the critical slope are functionals of the score, the first-best threshold and the agent's price of approval alone; the type distribution enters only the price of an insufficient slope budget (Proposition~\ref{prop:f-free}).

Two instances carry the analysis, AI agent oversight, where the principal's approval rule is itself the perturbation, and marketplace operation, where non-modular capacity makes revenue non-affine in the reported bids under sealed-bid execution; the same structure appears wherever a scored report is produced by the party the score is used to authorize, as with credit rating agencies and financial auditors.

\section*{Notation Summary}

Table~\ref{tab:notation} collects the notation of the paper; symbols local to the appendices are defined where they occur.

\begin{table}[tp]
\caption{Notation summary.}
\label{tab:notation}
\centering
\small
\renewcommand{\arraystretch}{0.92}
\begin{tabular}{lp{8cm}l}
\hline
Symbol & Meaning & Domain \\
\hline
$\Theta$, $\calR$, $\Omega$ & Type, report and outcome spaces & All \\
$S(r,\omega)$, $\Sbar$ & Scoring mechanism and its expected value (strictly proper) & All \\
$g(\theta)$ & Truthful report function & All \\
$h(r)$ & Perturbation payoff & All \\
$\gamma$ & Perturbation weight & All \\
$V(r;\theta,\gamma)$ & Reporter's combined objective & All \\
$r^*(\theta,\gamma)$ & Reporter's optimal report & All \\
$\mu$ & Common prior on $\Theta$ & All \\
$\calI = (\calY, \pi)$ & Information structure: signal space and signal law; enters only NT3 & All \\
$\Cbind$ & Binding conflict set (NT1): types with $h(g(\theta)) < \sup_r h(r)$ & All \\
$\calG$ & Credibility game & All \\
$G$, $\kappa = G''$ & Strictly convex generator of the score (Savage--McCarthy), and its curvature & Scoring rules \\
$D_G(p,r)$ & Bregman divergence of $G$: calibration loss of reporting $r$ at type $p$ & Scoring rules \\
$\kappa_{ij}$ & Non-modularity gap of the polymatroid capacity function & Market \\
$p$ & True success probability (agent type) & AI oversight \\
$q(r)$ & Approval rule & AI oversight \\
$\beta$ & Accuracy weight (price per unit of Bregman cost) & AI oversight \\
$p_{\min}$ & First-best threshold: the principal's break-even type & AI oversight \\
$u_s, u_f, u_d$ & Principal's utilities (success, failure, delegation) & AI oversight \\
$A$ & Stake of the decision, $A \coloneqq u_s - u_f > 0$ & AI oversight \\
$F$, $f$, $f_{\min}$ & Type distribution, its density, and its essential infimum & AI oversight \\
$C(r)$ & Marginal type's calibration cost, $D_G(p_{\min},r)$ & AI oversight \\
$r_0(G)$, $m$ & Reserve report, $C(r_0) = \gamma/\beta$, and the miscalibration it induces, $m = r_0(G) - p_{\min}$ & AI oversight \\
$\mathrm{ramp}(r)$, $s(r)$ & Bregman ramp $\min\{1,(\beta/\gamma)C(r)\}$ above $p_{\min}$, and its uncapped slope & AI oversight \\
$\Lambda$, $\zeta$ & Slope budget: the largest slope the approval rule may have; its relative shortfall $1 - \Lambda/\Lambda_c$ & AI oversight \\
$\Lambda_c(G)$ & Critical slope: the steepest chord of the ramp on the way to the reserve, the least slope budget attaining first best & AI oversight \\
$\mathcal Q_\Lambda$ & Monotone $\Lambda$-Lipschitz rules & AI oversight \\
$\delta_\Lambda(G)$ & Welfare gap under a slope budget $\Lambda$ & AI oversight \\
$\varkappa(G)$, $\xi(G)$ & Curvature ratio of~\eqref{eq:tilt-identity}; sharpening at the reserve & AI oversight \\
$\lambda_{\mathrm{sat}}$ & Approval intensity in the saturated regime $\gamma/\beta > D_G(p_{\min},1)$ & AI oversight \\
$W(q)$, $W^*$ & Principal's expected utility under $q$, and its first-best value & AI oversight \\
$\Pi(p)$ & Net gain to the principal from approving type $p$ & AI oversight \\
$\tau(p)$ & Induced screening, $q(r^*(p))$ & AI oversight \\
\hline
\end{tabular}
\end{table}

\begin{appendix}

\section{Proofs}\label{app:proofs}

Proofs follow in main-text order, after the envelope lemma, which every attainment proof is a hypothesis check against. Notation, the standing hypothesis (S), the assumption levels (A1)--(A3), the weak and strict feasibility conditions, the global-best-response convention and the two $D_G$ monotonicities are as set out in Section~\ref{sec:behavioral}, with $\partial_rD_G(p,r)=G''(r)(r-p)$ and $C'(r)=G''(r)(r-p_{\min})$ under (A2); each proof names the weakest hypothesis it needs, and every use of strict rather than weak feasibility is flagged where it occurs. Argmax ties are broken as Section~\ref{sec:behavioral} fixes them, which matters only in the two places that compare welfare across rules. Throughout, the \emph{pivot identity} $-\beta D_G(p,r)+\beta C(r)=-\beta\Phi(r;p)$, with $\Phi(r;p)\coloneqq G(p)-G(p_{\min})-G'(r)(p-p_{\min})$, follows under (A1) by expanding both divergences and cancelling the terms in $G(r)$ and $G'(r)r$; in $r$, $\Phi(\cdot\,;p)$ is strictly decreasing for $p>p_{\min}$, strictly increasing for $p<p_{\min}$, identically $0$ for $p=p_{\min}$, and $\Phi(p_{\min};p)=D_G(p,p_{\min})$. \suppnote

\subsection{Lemma~\ref{lem:envelope}: envelope sufficiency}\label{app:envelope}

\begin{proof}[Proof sketch]
The cap $q \le \mathrm{ramp}$ and the terminal agreement $q = 1 = \mathrm{ramp}$ on $[r_0,1]$ reduce the agent's problem to the envelope $\bar V(r;p) \coloneqq \gamma\,\mathrm{ramp}(r) - \beta D_G(p,r)$, whose shape in $r$ is read off the pivot identity, and the three type regimes are then checked in turn; the verification, for the extended form of the lemma that also covers the saturated regime, is in \suppsec{S.1}. Two pointwise consequences of the cap at a single report $r$ are used again below: $V(r;p) \le -\beta\Phi(r;p)$ for every $p$ by the pivot identity, so for $p < p_{\min}$ every report $r \ge p_{\min}$ with $q(r) \le \mathrm{ramp}(r)$ has $V(r;p) < 0$ while reports $r \le p_{\min}$ have $V(r;p) = -\beta D_G(p,r) \le 0$ with equality only at $r = p$ (the low-type branch), and for $p > p_{\min}$ every report $r \le p_{\min}$ has $V(r;p) = -\beta D_G(p,r) \le -\beta\Phi(p_{\min};p)$ (the low-report bound).
\end{proof}

\subsection{Lemma~\ref{lem:perturbation}: the residual types}\label{app:perturbation}

Parts~(i) and~(ii) are the classical layer, attributed in the lemma to \citet{boutilier2012eliciting} and proved by the sketch of Section~\ref{sec:fenchel}. Ours is the refinement below: the \emph{residual} types $\theta_0\in\Cbind$ with $\nabla h(g(\theta_0))=0$, on which (ii) is silent, are governed by a finite threshold, not excluded by a measure-zero argument.

\begin{proof}[Proof sketch]
The second-order expansion of $V$ at $g_0$ is governed by $M_\gamma\coloneqq H_S(\theta_0)+\gamma\nabla^2h(g_0)$, and two sub-cases exhaust it: if $\nabla^2h(g_0)$ has a positive eigenvalue then $M_\gamma$ acquires a positive direction at a finite $\gamma$ and $g_0$ stops being even a local maximizer; if $\nabla^2h(g_0)\preceq0$ then $g_0$ stays a local maximizer at every $\gamma$ and only a global comparison against an interior $r_1$ with $h(r_1)>h(g_0)$, available by NT1, decides, at the finite threshold $\Delta S/\Delta h$. The full verification is in \suppsec{S.2}.
\end{proof}

\subsection{Theorem~\ref{thm:optimal-nonaffine}}\label{app:optimal-oversight}

\begin{proof}[Proof sketch]
The cap of Lemma~\ref{lem:structure} forces an affine rule to vanish on $[0,p_{\min}]$, hence to vanish identically, hence to approve nobody, which is the regime-free screening half. Under weak feasibility the step rule of Theorem~\ref{thm:reserve}(i) attains $W^*$ while $\tau\Pi\leq\max(\Pi,0)$ caps every rule at $W^*$, so an optimum exists, attains first-best screening and is therefore not affine. The full verification is in \suppsec{S.3}.
\end{proof}

\subsection{Theorem~\ref{thm:reserve}}\label{app:reserve}

\begin{proof}[Proof sketch]
For~(i), $C$ is continuous and strictly increasing on $[p_{\min},1)$ from $C(p_{\min})=0$ with $\lim_{r\uparrow1}C(r)\geq\gamma/\beta$, so the intermediate value theorem gives a unique $r_0$, the limit rather than an endpoint value admitting the log score; the step rule at $r_0$ then satisfies the cap and the terminal agreement of Lemma~\ref{lem:envelope}. For~(ii), imposing $r_0=p_{\min}+c_0\sqrt{\gamma/\beta}$ across the whole feasible range forces $C(r)=(r-p_{\min})^2/c_0^2$, and integrating that identity makes the difference quotient of $G$ affine, hence $G$ quadratic; the converse is a computation. The full verification is in \suppsec{S.4}.
\end{proof}

\subsection{Lemma~\ref{lem:structure}}\label{app:structure}

\begin{proof}[Proof of the cap]
\emph{Assume} (A1) and let the measurable $q$ admit a global best response at every type and attain first-best screening, with $N\subset(0,p_{\min})$ the null set off which $\tau=0$. \emph{Prove} the cap, deriving weak feasibility on the way; the argument never uses where a low type best-responds, only that $\tau=0$ caps its payoff at zero.

\emph{Step 1 (cap).} For $p\in(0,p_{\min})\setminus N$ and any $r$, global optimality gives $-\beta D_G(p,r^*(p))+\gamma q(r^*(p))\geq-\beta D_G(p,r)+\gamma q(r)$, and the left side equals $-\beta D_G(p,r^*(p))\leq0$ because $\tau(p)=0$; so $\gamma q(r)\leq\beta D_G(p,r)$ for every $r$. Letting $p\uparrow p_{\min}$ along that full-measure set and using continuity of $p\mapsto D_G(p,r)$ gives $\gamma q\leq\beta C$ pointwise, whence $q(p_{\min})=0$, $q\leq\mathrm{ramp}$ on $[p_{\min},1]$, and $q<1$ on $[p_{\min},r_0)$. Taking instead $p\to r$ along the same set, for a fixed $r<p_{\min}$, gives $\gamma q(r)\leq\beta D_G(r,r)=0$ wherever $p \mapsto D_G(p,r)$ is continuous at $r$, which is every $r \in (0,p_{\min})$ and also $r = 0$ unless $G'(0^+) = -\infty$, so $q\equiv0$ on $(0,p_{\min})$ as well: a rule that pays anything at a report below $p_{\min}$ is bought by the types near that report, all of which lie below $p_{\min}$. Hence $q\leq\mathrm{ramp}$ on $(0,1]$, and on $[0,1]$ unless $G'(0^+)=-\infty$, in which case $D_G(p,0)=+\infty$ for every $p>0$ and the report $0$ is reachable by the type $p=0$ alone.

\emph{Step 2 (weak feasibility).} Weak feasibility follows instead of being assumed: some $p>p_{\min}$ has $q(r^*(p))=1$, and $r^*(p)>0$ since $V(0;p)=-\infty$ in the exceptional case, so $1\leq(\beta/\gamma)C(r^*(p))\leq(\beta/\gamma)C(1)$, that is $\gamma/\beta\leq D_G(p_{\min},1)$, and $r_0$ has a referent.

Each limit above is a limit of a pointwise inequality in a parameter entering $D_G$ continuously, so no continuity of $q$ is assumed. So both cap inequalities hold for every measurable $q$ with $\tau=0$ for $F$-a.e. $p<p_{\min}$, whether or not it attains first-best screening; Proposition~\ref{prop:lambda-sat} uses them in that weaker form, the saturated regime leaving $r_0$ and the saturation statement without content.
\end{proof}

\begin{proof}[Proof of the saturation statement]
Assume (A1) and weak feasibility and let $q$ attain first-best screening. Types in $(p_{\min},r_0)$ have positive measure, so pick such a $p$ with $\tau(p)=1$ and put $r_1\coloneqq r^*(p)$, so $q(r_1)=1$. The cap puts $q\equiv0$ on $(0,p_{\min}]$ and $r_1>0$ as in its proof, so $r_1>p_{\min}$. Then $1=q(r_1)\leq(\beta/\gamma)C(r_1)$ gives $C(r_1)\geq C(r_0)$, hence $r_1\geq r_0$; if $r_0=1$ this already gives $r_1=1$ and the statement. Otherwise take $p_n\downarrow p_{\min}$ with $\tau(p_n)=1$ and $r_1^n\coloneqq r^*(p_n)\in[r_0,1]$; optimality against $r=p_n$ gives $D_G(p_n,r_1^n)\leq\gamma/\beta=C(r_0)$. Fix $r'\in(r_0,1)$. Whenever $r_1^n\geq r'$, monotonicity of $r\mapsto D_G(p_n,r)$ above $p_n$ gives $D_G(p_n,r')\leq C(r_0)$, and $D_G(p_n,r')\to D_G(p_{\min},r')>C(r_0)$ by continuity at the interior report $r'$; so $r_1^n<r'$ for all large $n$, and since $r'>r_0$ was arbitrary, $r_1^n\to r_0$. Since $q<1$ on $[p_{\min},r_0)$ and $q\equiv0$ on $(0,p_{\min}]$, $\inf\{r>0:q(r)=1\}=r_0$, with $q=1$ at points accumulating at $r_0$ from the right, and continuity of $q$ upgrades this to $q(r_0)=1$.
\end{proof}

\subsection{Theorem~\ref{thm:kink}}\label{app:kink}

\begin{proof}[Proof of part~(i)]
$q_G = \mathrm{ramp}$ is non-decreasing, $[0,1]$-valued and, under (S), Lipschitz with constant $\sup_{[p_{\min},r_0]}s$, finite because $r_0<1$; it satisfies the cap and the terminal agreement of Lemma~\ref{lem:envelope} with equality throughout, which gives every claim of part~(i), the argmax at $p_{\min}$ being the whole interval $[p_{\min},r_0]$ because $q_G$ agrees with the ramp there. Comparing the two argmax descriptions shows that $q_G$ and the step rule of Theorem~\ref{thm:reserve} induce the same outcome for every type except $p_{\min}$.
\end{proof}

\begin{proof}[Proof of part~(ii)]
\emph{Assume} (S), that is (A2) and strict feasibility. \emph{Prove} part~(ii): smoothing the kink costs $O(\varepsilon)$ above $p_{\min}$ and $O(\varepsilon^2)$ below.

\emph{Step 1 (construction).} Fix a $C^\infty$ transition $\varphi:\R\to[0,1]$ with $\varphi\equiv0$ on $(-\infty,0]$, $\varphi\equiv1$ on $[1,\infty)$, $\varphi'\geq0$, $c_\varphi\coloneqq\|\varphi'\|_\infty\leq2$, and for $\varepsilon\in(0,m)$ put $q_\varepsilon(r)\coloneqq\mathrm{ramp}(r)+(1-\mathrm{ramp}(r))\varphi((r-r_0+\varepsilon)/\varepsilon)$. Then $q_\varepsilon=\mathrm{ramp}$ on $[0,r_0-\varepsilon]$, $q_\varepsilon\equiv1$ on $[r_0,1]$, $\mathrm{ramp}\leq q_\varepsilon\leq1$, and $q_\varepsilon'=\mathrm{ramp}'(1-\varphi)+(1-\mathrm{ramp})\varphi'/\varepsilon\geq0$. Every derivative vanishes at $r_0^-$, each Leibniz term carrying $\varphi^{(k)}(1)=0$ with $k\geq1$, or $(1-\varphi)(1)=0$, or $(1-\mathrm{ramp})(r_0)=0$; and $\mathrm{ramp}'(p_{\min})=0$ under (A2), so $q_\varepsilon\in C^1$ across $p_{\min}$, with $\|q_\varepsilon'\|_\infty\leq\sup s+c_\varphi(1-\mathrm{ramp}(r_0-\varepsilon))/\varepsilon\leq3\sup_{[p_{\min},r_0]}s$ uniformly in $\varepsilon$.

\emph{Step 2 (above $p_{\min}$).} Types $p\geq r_0$ contribute nothing, since $q_\varepsilon(p)=1$ and $V(r;p)\leq\gamma-\beta D_G(p,r)<\gamma=V(p;p)$ for $r\neq p$. Types $p\in(p_{\min},r_0)$ are the dominant channel: for $r\in[p_{\min},r_0-\varepsilon]$, $q_\varepsilon=\mathrm{ramp}=(\beta/\gamma)C$ and $V(r;p)=-\beta\Phi(r;p)<-\beta\Phi(r_0;p)=V(r_0;p)$ by the pivot identity, while for $r<p_{\min}$ the ramp vanishes and the low-report bound of Lemma~\ref{lem:envelope}'s proof gives $V(r;p)=-\beta D_G(p,r)\leq-\beta\Phi(p_{\min};p)<V(r_0;p)$; so $r^*(p)>r_0-\varepsilon$, hence either $\tau(p)=1$, or $r^*\in(r_0-\varepsilon,r_0)$ and optimality against the deviation $r=r_0$ gives, with $\bar G''_\varepsilon\coloneqq\max_{[r_0-\varepsilon,r_0]}G''$,
\[
\gamma(1-\tau(p))\leq\beta[D_G(p,r_0)-D_G(p,r^*)]=\beta\!\int_{r^*}^{r_0}\!G''(z)(z-p)\,dz\leq\beta\bar G''_\varepsilon\varepsilon(r_0-p),
\]
which integrated against the weight of~\eqref{eq:welfare-gap-identity}, with $f_{\max}\coloneqq\operatorname{ess\,sup}f<\infty$, contributes
\[
Af_{\max}(\beta/\gamma)\bar G''_\varepsilon\,\varepsilon\int_0^m x(m-x)\,dx=Af_{\max}\beta\bar G''_\varepsilon m^3\varepsilon/(6\gamma).
\]

\emph{Step 3 (below $p_{\min}$, and the sum).} Types $p<p_{\min}$ are subdominant: $q_\varepsilon=\mathrm{ramp}$ on $[0,r_0-\varepsilon]$ gives $V(r;p)<0$ there for $r\neq p$ by the low-type branch of Lemma~\ref{lem:envelope}'s proof, applicable because the cap holds with equality on that interval, and $V(r;p)\leq\gamma-\beta D_G(p,r_0)<0$ for $r\geq r_0$, so a false approval needs $\sup_{(r_0-\varepsilon,r_0)}V>0$, hence by $q_\varepsilon\leq1$ needs $\beta D_G(p,r_0-\varepsilon)<\gamma$, i.e. $\rho(p)>r_0-\varepsilon$ where $D_G(p,\rho(p))=\gamma/\beta$, i.e. $p_{\min}-p<c_\rho\varepsilon(1+o(1))$ with $c_\rho\coloneqq1/\rho'(p_{\min})$, the derivative of the \emph{inverse} map at $r_0$, since $\rho$ is increasing with $\rho(p_{\min})=r_0$ and the condition inverts to $p>\rho^{-1}(r_0-\varepsilon)$; implicit differentiation of $D_G(p,\rho(p))=\gamma/\beta$ gives $c_\rho=G''(r_0)(r_0-p_{\min})/[G'(r_0)-G'(p_{\min})]$ under (A2), which is $1$ under Brier and differs from $\mathrm{Lip}(\rho)$ whenever the curvature is not constant. Bounding $\tau\leq1$ there contributes $Af_{\max}c_\rho^2\varepsilon^2(1+o(1))/2$. So $W^*-W(q_\varepsilon)\leq Af_{\max}\beta\bar G''_\varepsilon m^3\varepsilon/(6\gamma)+Af_{\max}c_\rho^2\varepsilon^2(1+o(1))/2$, the infimum over $C^1$ rules is $0$, and by Corollary~\ref{cor:c1} it is not attained. If $G\in C^\infty$ the family may be taken in $C^\infty$ by multiplying the ramp by a smooth cutoff supported in $[p_{\min},1]$, which only lowers $q_\varepsilon$ below the ramp near $p_{\min}$.
\end{proof}

\subsection{Theorem~\ref{thm:critical-slope} and its corollaries}\label{app:critical-slope}

\begin{proof}[Proof of part~(i), the chord bound]
Assume (S) and let $q$ be continuous and attain first-best screening. By Lemma~\ref{lem:structure}, $q(r_0)=1$ and $q(a)\leq\mathrm{ramp}(a)$. So $q(r_0)-q(a)\geq1-\mathrm{ramp}(a)$ for every $a\in[p_{\min},r_0)$, and dividing by $r_0-a>0$ and taking the supremum gives $\mathrm{Lip}(q)\geq\Lambda_c$. Letting $a\uparrow r_0$, the right side tends to $(\beta/\gamma)C'(r_0)=s(r_0)$ under (A2), so every left difference quotient of $q$ at $r_0$ has $\liminf\geq s(r_0)$ and $q'(r_0^-)\geq s(r_0)>0$ where it exists. The condition is global: the same two facts bound every chord to $(r_0,1)$, not only the tangent, so the whole chord family binds.
\end{proof}

\begin{proof}[Proof of part~(ii), the threshold]
\emph{Sketch.} A clipped line settles attainment, and compactness turns attainment into the equivalence for the welfare gap.

\emph{Assume} (S). \emph{Prove} the two equivalences of part~(ii), for attainment and for the welfare gap.

\emph{Step 1 (attainment).} For the attainment direction, let $\Lambda\geq\Lambda_c$. Write $\mathcal{L}_\Lambda$ for the $\Lambda$-Lipschitz rules $[0,1]\to[0,1]$ with no monotonicity requirement, so that $\mathcal{Q}_\Lambda\subset\mathcal{L}_\Lambda$. Put $q_\Lambda(r)\coloneqq\operatorname{med}\{0,1-\Lambda(r_0-r),1\}$, a clipped affine function of slope $\Lambda$, hence non-decreasing, $[0,1]$-valued and $\Lambda$-Lipschitz; its zero is at $r_0-1/\Lambda$, and $a=p_{\min}$ in the definition of $\Lambda_c$ gives $\Lambda\geq\Lambda_c\geq1/m$, so $q_\Lambda\equiv0$ on $[0,p_{\min}]$. It lies under the ramp, since for $a\in[p_{\min},r_0)$ the definition of $\Lambda_c$ gives $1-\Lambda(r_0-a)\leq\mathrm{ramp}(a)$, and for $r\geq r_0$ both sides equal $1$. So $q_\Lambda$ satisfies the cap and the terminal agreement of Lemma~\ref{lem:envelope}, and the lemma gives first-best screening with a unique best response for every $p\neq p_{\min}$. With part~(i) this settles \emph{attainment} in both classes at once, since necessity is proved over $\mathcal{L}_\Lambda$ and the construction is monotone: a rule attaining first-best screening exists in $\mathcal{Q}_\Lambda$, equivalently in $\mathcal{L}_\Lambda$, if and only if $\Lambda\geq\Lambda_c(G)$. The threshold slope is admissible: $q_{\Lambda_c}$ satisfies both hypotheses by the same two inequalities at $\Lambda = \Lambda_c$, so it is first-best, although the supremum defining $\Lambda_c$ need not be attained by any anchor (under Brier it is approached only as $a \uparrow r_0$).

\emph{Step 2 (compactness).} The equivalence claimed for the welfare gap needs one step more, since $\delta_\Lambda(G)=0$ concerns a supremum and non-attainment alone does not force a positive gap, as Theorem~\ref{thm:kink}(ii) with Corollary~\ref{cor:c1} witnesses. Compactness closes it: both classes are uniformly bounded and uniformly $\Lambda$-Lipschitz, and the range $[0,1]$, the Lipschitz constant and (for $\mathcal{Q}_\Lambda$) monotonicity survive uniform limits, so both are compact in $C([0,1])$ by Arzel\`a--Ascoli \citep[Thm.~3.85]{aliprantis2006infinite}.

\emph{Step 3 (the welfare gap).} Under (A1) $r\mapsto D_G(p,r)$ is continuous into $[0,+\infty]$, so $V(\cdot\,;p)$ is upper semicontinuous on the compact report space with non-empty compact argmax, while $|\max_rV(r;p)-\max_rV'(r;p)|\leq\gamma\|q-q'\|_\infty$ uniformly in $p$; Berge's maximum theorem \citep[Thm.~17.31]{aliprantis2006infinite} then makes the best-response correspondence upper hemicontinuous in $q$, so along $q_n\to q$ every limit of maximizers maximizes the limit problem. Hence $q\mapsto\max\{q(r)\Pi(p):r\text{ a best response at }p\}$ is upper semicontinuous for each $p$, and, its integral in~\eqref{eq:principal-utility} being dominated by $Af$, reverse Fatou makes the welfare functional $\widehat W$ that takes this selection everywhere upper semicontinuous on each class. Section~\ref{sec:behavioral} selects the largest approval among best responses, which is the principal's preference above $p_{\min}$ but not below it, whereas $\widehat W$ takes the principal's preference at every type; so $\widehat W$ dominates $W$ under that selection and under every other, and being upper semicontinuous on a compact set it attains its supremum, at some $q^\dagger$; were $\widehat W(q^\dagger)=W^*$ then~\eqref{eq:welfare-gap-identity} with $A>0$ and $f>0$ would make $q^\dagger$ attain first-best screening, and $q^\dagger$ is Lipschitz hence continuous, so Lemma~\ref{lem:structure}, whose proof uses only the values of $\tau$ and so holds under any selection, and part~(i) would give $\Lambda\geq\mathrm{Lip}(q^\dagger)\geq\Lambda_c(G)$. So, writing $\delta^{\mathcal{L}}_\Lambda(G)$ for the gap over $\mathcal{L}_\Lambda$, $\delta_\Lambda(G)=\delta^{\mathcal{L}}_\Lambda(G)=0$, under any selection, if and only if $\Lambda\geq\Lambda_c(G)$, and part~(ii) rests on (S) alone and not on Proposition~\ref{prop:below-threshold}; this also disposes of the case in which the chord supremum exceeds $s(r_0)$ (\suppsec{S.9}), a case the proposition leaves open.
\end{proof}

\begin{proof}[Proof of Corollary~\ref{cor:c1}]
Assume (S) and suppose $q\in C^1([0,1],[0,1])$ attains first-best screening. Strict feasibility puts $r_0$ in $(p_{\min},1)$, and Lemma~\ref{lem:structure}, applicable since $q$ is continuous, gives $q(r_0)=1$; as $q\leq1$ and $r_0$ is interior, $r_0$ is an interior maximizer and $q'(r_0)=0$, while the chord bound of Theorem~\ref{thm:critical-slope}(i), as $a\uparrow r_0$, gives $q'(r_0^-)\geq s(r_0)>0$. Differentiability at $r_0$ alone suffices, by a chain that bypasses the continuity hypothesis: the cap needs only (A1) and gives $q\leq\mathrm{ramp}<1$ below $r_0$; the saturation statement without continuity supplies $r_n\downarrow r_0$ with $q(r_n)=1$, and differentiability at $r_0$ gives continuity there, so $q(r_0)=1$; then every left difference quotient at $r_0$ is at least the corresponding chord of the ramp, whose limit as $a\uparrow r_0$ is $s(r_0)>0$, while every right one is at most $0$, so $q$ is not differentiable at $r_0$. The infimum of Theorem~\ref{thm:kink}(ii) is therefore not attained.
\end{proof}

\begin{proof}[Proof sketch of Corollary~\ref{cor:first-best}]
Necessity is Lemma~\ref{lem:structure} applied to a type in $(p_{\min},r_0)$ that is approved, whose best response must land at $r_0$; sufficiency is that $q\equiv1=\mathrm{ramp}$ on $[r_0,1]$, so the cap and the terminal agreement of Lemma~\ref{lem:envelope} hold. The verification is in \suppsec{S.8}.
\end{proof}

\subsection{Proposition~\ref{prop:below-threshold}}\label{app:below-threshold}

\begin{proof}
\emph{Sketch.} A band of types above $p_{\min}$ stops short of the reserve; the upper bound prices this dominant channel, and the lower bound adds false approvals below $p_{\min}$.

\emph{Assume} (A3) and strict feasibility, which places $r_0$ in $(p_{\min},1)$, and let $s$ be non-decreasing on $[p_{\min},1)$ with $s'(r_0)>0$; by the closed form of Section~\ref{sec:critical}, $\Lambda_c=s(r_0)$. \emph{Prove} both bounds and the ratio of their constants.

\emph{Step 1 (the stopping barrier).} Write $x\coloneqq p-p_{\min}$, $\Lambda=\Lambda_c(1-\zeta)$, $K\coloneqq\Lambda_c^2/s'(r_0)$, $\Gamma_0\coloneqq G'(r_0)-G'(p_{\min})$, $b\coloneqq(\beta/\gamma)\Gamma_0$, and $s_p(z)\coloneqq(\beta/\gamma)G''(z)(z-p)$, so that $s_{p_{\min}}=s$ and $\beta[D_G(p,r')-D_G(p,r)]=\gamma\int_r^{r'}s_p$. Since $s_p(z)=s(z)\cdot(z-p)/(z-p_{\min})$ is a product of non-negative non-decreasing functions of $z$ on $(p,1)$, the hypothesis that $s$ is non-decreasing makes $s_p$ non-decreasing there. For $x\in(0,m\zeta)$ this makes
\begin{equation}\label{eq:stopping-barrier}
r_c(p)\coloneqq\min\{r:s_p(r)\geq\Lambda\}
\end{equation}
well defined and interior to $(p,r_0)$, since $s_p(p)=0<\Lambda<\Lambda_c(1-x/m)=s_p(r_0)$, and puts $s_p\geq\Lambda$ on all of $[r_c(p),1)$. Fix $\zeta_0$ small enough that $s_p$ is strictly increasing on $[r_c(p),r_0]$ throughout the band, which $s'(r_0)>0$ and the continuity of $s'$ under (A3) supply uniformly. Then for every $q\in\mathcal{L}_\Lambda$, absolutely continuous with $|q'|\leq\Lambda$ a.e., and every $r'>r_c(p)$, $V(r';p)-V(r_c(p);p)=\gamma\int_{r_c(p)}^{r'}[q'-s_p]\leq\gamma\int_{r_c(p)}^{r'}[\Lambda-s_p]<0$. So $r^*(p)\leq r_c(p)$: no rule in $\mathcal{L}_\Lambda$, monotone or not, moves a band type past $r_c(p)$, and the barrier is global.

\emph{Step 2 (upper bound).} Let $\hat q_\Lambda(r)\coloneqq\min\{1,\inf_{a\in[p_{\min},1]}[\mathrm{ramp}(a)+\Lambda|r-a|]\}$ for $r\geq p_{\min}$ and $0$ below: the largest $\Lambda$-Lipschitz function under the ramp, non-decreasing and at most $\mathrm{ramp}$, hence in $\mathcal{Q}_\Lambda$, so by the low-type branch of Lemma~\ref{lem:envelope}'s proof it produces no false approvals at all. Because $s$ is non-decreasing it has three pieces: $\mathrm{ramp}$ on $[p_{\min},r_t]$ where $s(r_t)=\Lambda$, then the line of slope $\Lambda$ through $(r_t,\mathrm{ramp}(r_t))$, then $1$ from $R_\Lambda\coloneqq\inf\{r:\hat q_\Lambda(r)=1\}=r_0+\eta/\Lambda$ on, with $\eta\coloneqq\sup_a[(1-\mathrm{ramp}(a))-\Lambda(r_0-a)]>0$ the shortfall at the reserve; expanding $1-\mathrm{ramp}(r_0-h)=\int_{r_0-h}^{r_0}s$ around $r_0$ gives $\eta=\Lambda_c^2\zeta^2/(2s'(r_0))(1+o(1))$, quadratically small in the deficit. Under this rule the barrier holds with equality, the agent's problem being unimodal: $\partial_rV=\beta G''(r)x>0$ on $(p_{\min},r_t)$, $\partial_rV=\gamma[\Lambda-s_p(r)]$ on $(r_t,R_\Lambda)$, positive exactly below $r_c(p)$, $\partial_rV=-\gamma s_p(r)<0$ above $R_\Lambda$, and $-\beta D_G(p,\cdot)$ increasing on $[0,p_{\min})$; so $r^*(p)=r_c(p)$ uniquely. On the terminal segment the inf-convolution is the line of slope $\Lambda$ landing on $1$ at $R_\Lambda$, so $1-\tau(p)=\Lambda(R_\Lambda-r^*(p))$. Writing $r_c(p)=r_0-y(x)$ and solving $s_p(r_c)=\Lambda$ to first order gives $y(x)=\Lambda_c(\zeta-x/m)/s'(r_0)+O(\zeta^2)$, so the affected band is $x\in(0,m\zeta)$ and on it $1-\tau(p)=K(\zeta-x/m)(1+O(\zeta))$. Integrating against the weight of~\eqref{eq:welfare-gap-identity},
\[
\delta_\Lambda(G)\leq W^*-W(\hat q_\Lambda)=Af(p_{\min})\frac{\Lambda_c^2}{s'(r_0)}\int_0^{m\zeta}\!x\Bigl(\zeta-\frac{x}{m}\Bigr)dx\,(1+O(\zeta))=\frac{Af(p_{\min})\Lambda_c^2m^2}{6\,s'(r_0)}\zeta^3(1+o(1)),
\]
the $o(1)$ absorbing the $O(\zeta)$ remainders of the expansion in $x$ and the modulus of continuity of $f$ at $p_{\min}$, which the continuous density of Definition~\ref{def:oversight-game} does not rate. For the Brier generator every step is exact: $s'\equiv2\beta/\gamma$, $\Lambda_c=2\beta m/\gamma$, $\eta=(\beta/\gamma)m^2\zeta^2$, $y(x)=m\zeta-x$, and $1-\tau(p)=(2\zeta-\zeta^2)-2(1-\zeta)x/m$ on the band, which vanishes at $x_1=m(2\zeta-\zeta^2)/(2(1-\zeta))$; so for $F$ uniform $W^*-W(\hat q_\Lambda)=Af(\beta m^4/\gamma)(2\zeta-\zeta^2)^3/[24(1-\zeta)^2]=Af(\beta m^4/\gamma)[\zeta^3/3+\zeta^4/6+O(\zeta^5)]$, exact rather than a two-term truncation.

\emph{Step 3 (lower bound, shifted cap).} Fix $q\in\mathcal{L}_\Lambda$ and let $w$ range over $(0,\bar w]$ for a fixed $\bar w<p_{\min}$, small enough that $G'$ stays finite on $[p_{\min}-\bar w,r_0]$. The proof rests on the envelope inequality $\gamma q(r)-\beta D_G(p,r)\leq\gamma\tau(p)$, valid for all $p,r$ and every $q$ because the left side is $V(r;p)\leq V(r^*(p);p)=\gamma\tau(p)-\beta D_G(p,r^*(p))\leq\gamma\tau(p)$. At $p=p_{\min}-w$ it gives, for $r\in[p_{\min},r_0]$, the shifted cap $q(r)\leq\mathrm{ramp}(r)+\tau(p_{\min}-w)+bw(1+O(w))$, since $(\beta/\gamma)[D_G(p_{\min}-w,r)-D_G(p_{\min},r)]=(\beta/\gamma)\int_{p_{\min}-w}^{p_{\min}}[G'(r)-G'(u)]\,du\leq(\beta/\gamma)w[G'(r_0)-G'(p_{\min}-w)]$. Put $\epsilon\coloneqq\inf_{w}[\tau(p_{\min}-w)+bw(1+O(w))]$, so that $q\leq\mathrm{ramp}+\epsilon$ on $[p_{\min},r_0]$ while $\tau(p_{\min}-w)\geq\epsilon-bw(1+O(w))$ for every $w$. The two halves of~\eqref{eq:welfare-gap-identity} are then charged separately and no case split on $q$ is needed.

\emph{Step 4 (dominant channel).} For a band type $p=p_{\min}+x$ with $r^*(p)\geq p_{\min}$, the barrier~\eqref{eq:stopping-barrier} and the monotonicity of $\mathrm{ramp}$ give $\tau(p)\leq\mathrm{ramp}(r_c(p))+\epsilon$, while $1-\mathrm{ramp}(r_c(p))=\int_{r_c(p)}^{r_0}s\geq\Lambda\,y(x)$ because $s\geq s_p\geq\Lambda$ there. The identity $s_p(r_0)-\Lambda=\Lambda_c(\zeta-x/m)$ is exact and $s_p(r_0)-s_p(r_0-y)\leq y\,s'(r_0)(1+O(\zeta))$, so $y(x)\geq\Lambda_c(\zeta-x/m)(1-O(\zeta))/s'(r_0)$ and $1-\tau(p)\geq K(\zeta-x/m)(1-O(\zeta))-\epsilon$. The remaining case is dispatched by the same envelope inequality, at $p_{\min}-w=r^*(p)$: if $r^*(p)<p_{\min}$ then $\tau(p)\leq\tau(r^*(p))$, and either $\tau(p)\leq\tfrac12$, when the display holds because $1-\tau(p)\geq\tfrac12>K\zeta$, or $q>\tfrac12$ at a report below $p_{\min}$, when the Lipschitz bound and $\tau\geq q$ there put $\tau\geq\tfrac14$ on an interval of length $\ell\coloneqq\min\{1/(4\Lambda),p_{\min}\}$ inside $[0,p_{\min})$, a false-approval term above $Af_{\min}\ell^2/8$ that is independent of $\zeta$ and dominates the proposition's bound for $\zeta<\zeta_0$. Integrating $Af_{\min}x$ against the positive part over the band therefore contributes at least $Af_{\min}Km^2\hat\zeta^3(1-O(\zeta))/6$, with $\hat\zeta\coloneqq(\zeta-\epsilon/K)_+$.

\emph{Step 5 (subdominant channel, and the sum).} By~\eqref{eq:welfare-gap-identity} the types below $p_{\min}$ contribute at least $Af_{\min}\int_0^{\bar w}w\,\tau(p_{\min}-w)\,dw\geq Af_{\min}\int_0^{\epsilon/b}w(\epsilon-bw)\,dw\,(1-O(\zeta))=Af_{\min}\epsilon^3(1-O(\zeta))/(6b^2)$, a step that uses no monotonicity of $q$ or of $\tau$, so the lower bound may quantify over $\mathcal{L}_\Lambda$. Put $\Gamma_1\coloneqq K/(K+bm)=\gamma\Lambda_c^2/[\gamma\Lambda_c^2+\beta\,\Gamma_0\,s'(r_0)\,m]\in(0,1)$ and write $\epsilon=tK\zeta$ with $t\geq0$, so that $K^3/b^2=Km^2\Gamma_1^2/(1-\Gamma_1)^2$. The two contributions sum to
\[
\frac{Af_{\min}Km^2\zeta^3}{6}\Bigl[(1-t)_+^3+\frac{\Gamma_1^2}{(1-\Gamma_1)^2}\,t^3\Bigr](1-O(\zeta))\;\geq\;\frac{Af_{\min}\Lambda_c^2m^2}{6\,s'(r_0)}\,\Gamma_1^2\,\zeta^3\,(1-O(\zeta)),
\]
the bracket being minimized over $t\geq0$ at $t=1-\Gamma_1$, where it equals $\Gamma_1^2$. The rule determined $\epsilon$ and hence $t$, and the bound is uniform in $t$, so it holds for every $q\in\mathcal{L}_\Lambda$ and all $\zeta<\zeta_0$: this is the lower bound, and dividing the upper constant by it gives the stated ratio $f(p_{\min})/(f_{\min}\Gamma_1^2)$. The factor $\Gamma_1$ weighs how fast the band loses approval against how fast the cap tilts as the anchoring type falls below $p_{\min}$, and through it alone the lower constant depends on the generator beyond $(\Lambda_c,m,s'(r_0))$. Under Brier $\Gamma_0=2m$, $\Lambda_c=2\beta m/\gamma$ and $s'(r_0)=2\beta/\gamma$ give $\Gamma_1=\tfrac12$, so the lower constant is $Af_{\min}\Lambda_c^2m^2/(24\,s'(r_0))$ and the ratio is $4$ for uniform $F$.
\end{proof}

\subsection{Proposition~\ref{prop:score-selection}}\label{app:score-selection}

Proposition~\ref{prop:score-selection} reads off Proposition~\ref{prop:score-selection-full}, stated and proved in Appendix~\ref{app:score-selection-full}: part~(ii) there is the distribution-freeness of the ranking, part~(iii) the two ratios against the Brier score, and part~(i) the closed forms of both constants that the main text leaves unstated.

\subsection{Proposition~\ref{prop:lambda-sat}}\label{app:lambda-sat}

\begin{proof}[Proof sketch]
$q_{\mathrm{sat}}$ is a hypothesis check against the extended form of Lemma~\ref{lem:envelope} in \suppsec{S.1}, with the reserve replaced by $1$, and optimality within the stated class follows from the two cap inequalities of Lemma~\ref{lem:structure} in the weaker form noted in its proof; the verification is in \suppsec{S.5}.
\end{proof}

\subsection{The Unrestricted Benchmark}\label{app:benchmark}

In the unrestricted direct mechanism the principal chooses a decision rule $x:[0,1]\to[0,1]$ and a transfer $t:[0,1]\times\{0,1\}\to\R$, and a type-$p$ agent reporting $r$ receives $u(r,p)=\gamma x(r)+\bar t(r;p)$ with $\bar t(r;p)=p\,t(r,1)+(1-p)t(r,0)$. The mechanism is \emph{incentive compatible} if $u(p,p)\geq u(r,p)$ for all $p,r$, strictly so if the inequality is strict for $r\neq p$. The principal's payoff over this class is the same decision payoff~\eqref{eq:principal-utility}, with $\tau=x$ on path; transfers are welfare-neutral and no participation or budget constraint is imposed, so the benchmark class is if anything too large and both comparisons below are conservative. The term $\gamma x(r)$ is not an instrument: the autonomy payoff attaches to the decision itself, so the principal can offset it but not switch it off.

\begin{theorem}[Canonical form of every incentive-compatible direct mechanism]\label{thm:canonical-form}
$(x,t)$ is incentive compatible if and only if there is a convex $\Psi:[0,1]\to\R$, with $\Psi'$ a selection of its subdifferential, such that
\begin{equation}\label{eq:canonical-form}
t(r,\omega)=\underbrace{\Psi(r)+\Psi'(r)(\omega-r)}_{\text{proper score with generator }\Psi}-\underbrace{\gamma\,x(r)}_{\text{approval charge}} .
\end{equation}
The rule $x$ is unrestricted, $\Psi$ is unrestricted among convex functions admitting a finite subgradient at every point of $[0,1]$, the two are chosen independently, and the agent's on-path payoff is $\Psi(p)$.
\end{theorem}

\begin{proof}[Proof sketch]
This is the standard envelope characterization of incentive compatibility in a one-dimensional linear environment \citep{myerson1981optimal,rochet1987necessary,vohra2011mechanism}, carried out at every pair of types rather than almost everywhere; the verification is in \suppsec{S.6}.
\end{proof}

Two coordinate readings price the restriction, written in the coordinates $B(r)\coloneqq t(r,1)-t(r,0)$, $\alpha(r)\coloneqq t(r,0)+\gamma x(r)$ and $U_0\coloneqq\max_ru(r,0)$ of \suppsec{S.6}, in which $u(r,p)=\alpha(r)+pB(r)$. Taking $\Psi(r)=\Psi_0+r^2/2$ makes the first bracket a scaled Brier score and leaves $x$ free, so every decision rule, the first-best included, is truthfully and strictly implementable and the unrestricted principal attains $W^*$. The restricted architecture is the sub-family of~\eqref{eq:canonical-form} with two deletions, of which exactly one binds: freezing the generator at $\Psi=\beta G$ is innocuous, since $t(r,\omega)=\beta[G(r)+G'(r)(\omega-r)]-\gamma x(r)+c$ has $B=\beta G'$ non-decreasing and $\alpha(r)=\beta G(r)-r\beta G'(r)+c$, which is~\eqref{eq:canonical-form} with $\Psi=\beta G+c$; whereas deleting the approval charge is fatal, since with $t(r,\omega)=\Psi_B(r)+\Psi_B'(r)(\omega-r)$ one has $\alpha(r)=\Psi_B(r)-r\Psi_B'(r)+\gamma x(r)$, and incentive compatibility demands $\alpha(r)=U_0+\int_0^r\Psi_B'-r\Psi_B'(r)$, forcing $\gamma x(r)$ constant. Truthfulness here requires the approval charge; redesigning the score is neither necessary nor sufficient for it.

\begin{theorem}[Welfare equivalence, truthfulness inequivalence]\label{thm:welfare-equivalence}
Fix the environment, a strictly convex $G\in C^2$, and $\beta,\gamma>0$. \emph{(i)} The unrestricted benchmark attains $W^*$ with a strictly incentive-compatible mechanism, and $W^*$ bounds every mechanism in either class. \emph{(ii)} The restricted architecture attains $W^*$ if and only if $\gamma/\beta\leq D_G(p_{\min},1)$, and when it does, both the step rule of Theorem~\ref{thm:reserve} and the Bregman ramp of Theorem~\ref{thm:kink} attain it. \emph{(iii)} No member of the restricted class attaining $W^*$ is incentive compatible; the incentive-compatible members are exactly the constant rules, all strictly below $W^*$.
\end{theorem}

\begin{proof}[Proof sketch]
(i) is the first coordinate reading above with the pointwise bound in~\eqref{eq:principal-utility}, (ii) combines Theorem~\ref{thm:reserve} with the cap of Lemma~\ref{lem:structure}, and (iii) turns on a quadratic H\"older condition that forces $q$ constant on $(0,1)$; the verification is in \suppsec{S.7}.
\end{proof}

The truthfulness-collapse step in~(iii) strengthens the specialization attributed in Lemma~\ref{lem:perturbation}(i) from every $\gamma>0$ to a fixed $\gamma$, which the quadratic H\"older bound of \suppsec{S.7} supplies and the classical statement does not. The two architectures are thus welfare-equivalent and truthfulness-inequivalent on the whole feasible region.

\subsection{Score Selection: The Full Statement}\label{app:score-selection-full}

Proposition~\ref{prop:score-selection} of the main text reads off parts~(ii) and~(iii) of the following statement, in the notation of Section~\ref{sec:critical}.

\begin{proposition}[Score selection below the critical slope, full form]\label{prop:score-selection-full}
Assume the hypotheses of Proposition~\ref{prop:below-threshold} and write $\kappa \coloneqq G''$. At the reserve put $\varkappa(G) \coloneqq \kappa(r_0)/\langle\kappa\rangle_\chi$, the curvature ratio of~\eqref{eq:tilt-identity}, $\xi(G) \coloneqq m\,\kappa'(r_0)/\kappa(r_0)$, the sharpening of the score at the reserve, which satisfies $1+\xi > 0$ because $s'(r_0) > 0$, and $\bar\kappa(G) \coloneqq (G'(r_0)-G'(p_{\min}))/m$, the mean curvature on $[p_{\min},r_0]$. Compare scores at a fixed induced miscalibration $m$, so that the prize $\gamma/\beta = D_G(p_{\min},r_0)$ varies with the score. Then:
\begin{enumerate}
\item[(i)] the upper constant of Proposition~\ref{prop:below-threshold} is $A f(p_{\min})\,\varkappa\,m^2/(3(1+\xi))$ and the lower constant is $A f_{\min}\,\varkappa\,m^2\,\Gamma_1^2/(3(1+\xi))$ with $\Gamma_1 = \varkappa/\bigl(\varkappa + (1+\xi)\bar\kappa/\langle\kappa\rangle_\chi\bigr)$;
\item[(ii)] the type distribution enters the upper constant only through $f(p_{\min})$ and the lower constant only through $f_{\min}$, both common to every score at the same $p_{\min}$, so the ranking of scores by either constant at a given relative shortfall $\zeta$ is distribution-free;
\item[(iii)] relative to the Brier score at the same $m$, a score's upper constant is $\varkappa/(1+\xi)$ times Brier's while its critical slope is $\varkappa$ times Brier's; if $\kappa$ is convex on $[p_{\min},r_0]$ and $0 \le \xi \le 2$, then $\varkappa \le 1+\xi$, with equality if and only if $\kappa$ is constant on $[p_{\min},r_0]$.
\end{enumerate}
\end{proposition}

\begin{proof}
Under the hypotheses $\Lambda_c=s(r_0)=(\beta/\gamma)\kappa(r_0)m$ and $s'(r_0)=(\beta/\gamma)\bigl(\kappa'(r_0)m+\kappa(r_0)\bigr)=(\beta/\gamma)\kappa(r_0)(1+\xi)$, so $s'(r_0)>0$ is $1+\xi>0$. Feasibility, $C(r_0)=\gamma/\beta$, together with $C(r_0)=\int_{p_{\min}}^{r_0}\kappa(t)(t-p_{\min})\,dt=\langle\kappa\rangle_\chi m^2/2$, gives $\beta/\gamma=2/(\langle\kappa\rangle_\chi m^2)$, hence $\Lambda_c=2\varkappa/m$ and $s'(r_0)=2\kappa(r_0)(1+\xi)/(\langle\kappa\rangle_\chi m^2)$. Substituting into the upper constant $A f(p_{\min})\Lambda_c^2m^2/(6s'(r_0))$ of Proposition~\ref{prop:below-threshold} gives $A f(p_{\min})\varkappa m^2/(3(1+\xi))$. For the lower constant, $K=\Lambda_c^2/s'(r_0)=2\varkappa/(1+\xi)$ and $bm=(\beta/\gamma)\Gamma_0 m=2\bar\kappa/\langle\kappa\rangle_\chi$, so $\Gamma_1=K/(K+bm)$ is as stated and $A f_{\min}\Lambda_c^2m^2\Gamma_1^2/(6s'(r_0))=A f_{\min}\varkappa m^2\Gamma_1^2/(3(1+\xi))$; this is~(i), and~(ii) reads it off. For~(iii), the Brier score has $\varkappa=1$ and $\xi=0$, which gives both ratios. If $\kappa$ is convex on $[p_{\min},r_0]$ it lies above its tangent at $r_0$, so $\kappa(r_0)-\kappa(t)\le\kappa'(r_0)(r_0-t)$ there; averaging against $\chi$, whose mean of $r_0-t$ is $m/3$, gives $\kappa(r_0)-\langle\kappa\rangle_\chi\le\kappa'(r_0)m/3$, that is $\varkappa-1\le\xi\varkappa/3$. For $0\le\xi<3$ this rearranges to $\varkappa\le3/(3-\xi)$, and $3/(3-\xi)\le1+\xi$ exactly when $\xi(2-\xi)\ge0$, i.e.\ $\xi\le2$. Equality in the tangent step forces $\kappa(t)=\kappa(r_0)+\kappa'(r_0)(t-r_0)$ $\chi$-almost everywhere, so $\kappa$ is affine on the interval, and positivity of $\kappa$ at $p_{\min}$ then gives $\xi<1$; for an affine $\kappa$ the first inequality is an equality, $\varkappa=3/(3-\xi)$, which is strictly below $1+\xi$ for $0<\xi<1$ and equal to it only at $\xi=0$, i.e.\ only when $\kappa$ is constant.
\end{proof}

\subsection{Lemma 4.4: envelope sufficiency, extended form}

\paragraph*{Lemma 4.4 (Envelope sufficiency, extended form).}
Assume (A1). Write $\bar r \coloneqq \min\{r_0,1\}$, taking $\bar r \coloneqq 1$ in the saturated regime $\gamma/\beta > D_G(p_{\min},1)$ where no reserve exists, and put $\bar q \coloneqq \mathrm{ramp}(\bar r)$, so $\bar q = 1$ under weak feasibility and $\bar q = \lambda_{\mathrm{sat}} \coloneqq \beta D_G(p_{\min},1)/\gamma$ in the saturated regime. Let $q : [0,1] \to [0,1]$ satisfy
\begin{enumerate}
\item[(C1)] \emph{(cap)} $q \le \mathrm{ramp}$ on $[0,\bar r]$, and
\item[(C2)] \emph{(terminal agreement)} $q = \mathrm{ramp}$ on $[\bar r,1]$.
\end{enumerate}
Then $\tau(p) = \bar q\,\mathbf 1\{p \ge p_{\min}\}$ for every $p \ne p_{\min}$, with a unique best response, namely $r^*(p) = p$ for $p < p_{\min}$ and $r^*(p) = \max(p,\bar r)$ for $p > p_{\min}$. At $p = p_{\min}$ the argmax is $\{p_{\min}\} \cup \{r \in [p_{\min},\bar r] : q(r) = \mathrm{ramp}(r)\}$, a single ($F$-null) type. No monotonicity, continuity or semicontinuity of $q$ is assumed.

The main paper states the case $\bar r = r_0$, $\bar q = 1$ under weak feasibility; the form here adds the saturated regime, which Proposition~4.15 uses.

\begin{proof}
\emph{Sketch.} The cap bounds the agent's payoff by an envelope that depends on the rule only through the ramp, the terminal agreement makes the bound tight on $[\bar r,1]$, and the pivot identity gives the envelope's monotonicity in the report on $[p_{\min},\bar r]$; the three type regimes are then read off.

\emph{Assume} (A1), (C1) and (C2). \emph{Prove} the three regime statements, with the uniqueness of the best response off $p_{\min}$.

\emph{Step 1 (the envelope).} By the pivot identity of Appendix~A and (C1), $V(r;p) \le \gamma\,\mathrm{ramp}(r) - \beta D_G(p,r) \eqqcolon \bar V(r;p)$ for every $r$, with equality at every $r \in [\bar r,1]$ by (C2); and $\bar V(r;p) = -\beta\Phi(r;p)$ on $[p_{\min},\bar r]$ and $\gamma\bar q - \beta D_G(p,r)$ on $[\bar r,1]$, where $\gamma \bar q = \beta C(\bar r)$ in both regimes. Since $\mathrm{ramp} \equiv 0$ on $[0,p_{\min}]$, (C1) forces $q \equiv 0$ there.

\emph{Step 2 (types above $p_{\min}$).} On $[p_{\min},\bar r]$, $\Phi(\cdot\,;p)$ is strictly decreasing, so $V(r;p) \le \bar V(r;p) < \bar V(\bar r;p) = V(\bar r;p)$ for $r < \bar r$; on $[0,p_{\min}]$, $V(r;p) = -\beta D_G(p,r) \le -\beta D_G(p,p_{\min}) = -\beta\Phi(p_{\min};p) < V(\bar r;p)$; and on $[\bar r,1]$, $V(r;p) = \gamma\bar q - \beta D_G(p,r)$ is maximized uniquely at $r = \max(p,\bar r)$. So $r^*(p) = \max(p,\bar r)$ uniquely and $\tau(p) = \bar q$.

\emph{Step 3 (types below $p_{\min}$).} On $[p_{\min},\bar r]$, $\Phi(\cdot\,;p)$ is strictly increasing, so $V \le \bar V(p_{\min};p) = -\beta D_G(p,p_{\min}) < 0$; on $[\bar r,1]$, $V \le \gamma\bar q - \beta D_G(p,\bar r) < \gamma\bar q - \beta D_G(p_{\min},\bar r) = 0$, since $p \mapsto D_G(p,\bar r)$ is strictly decreasing below $\bar r$ (for $\bar r < 1$ by Section~4.1; for $\bar r = 1$, which arises only when $D_G(p_{\min},1) < \infty$ and hence $G'(1^-)$ is finite, by $\partial_p D_G(p,1) = G'(p) - G'(1^-) < 0$); on $[0,p_{\min}]$, $V = -\beta D_G(p,r) \le 0$ with equality only at $r = p$. So $r^*(p) = p$ uniquely and $\tau(p) = 0$.

\emph{Step 4 (the marginal type).} $\Phi(\cdot\,;p_{\min}) \equiv 0$, so $\bar V \equiv 0$ on $[p_{\min},\bar r]$ and $\bar V < 0$ elsewhere off $r = p_{\min}$, and $V \le \bar V$ gives the stated argmax. No strict feasibility is used; in the saturated regime $[\bar r,1]$ degenerates to the point $r = 1$.
\end{proof}

\subsection{Lemma 3.1(iii): the residual types}

\paragraph*{Lemma 3.1(iii) (Perturbation Lemma, residual types, restated).}
\emph{(Residual types.)} If $\nabla h(g(\theta_0)) = 0$ for some $\theta_0 \in \Cbind$ with $h(g(\theta_0)) < \sup_r h(r)$, and $h \in C^2$ near $g(\theta_0)$, then there is a type-dependent threshold $\bar\gamma(\theta_0) > 0$ such that $g(\theta_0)$ fails to maximize $V(\cdot\,;\theta_0,\gamma)$ for every $\gamma > \bar\gamma(\theta_0)$.

\begin{proof}[Proof of Lemma~3.1(iii)]
Assume $h\in C^2$ near $g_0\coloneqq g(\theta_0)$, a hypothesis of part~(iii) and not a consequence of NT2, and NT1 at $\theta_0$, so $h(g_0)<\sup_rh$. Local behavior is governed by $M_\gamma\coloneqq H_S(\theta_0)+\gamma\nabla^2h(g_0)$ through $V(r;\theta_0,\gamma)=V(g_0;\theta_0,\gamma)+\tfrac12(r-g_0)^\top M_\gamma(r-g_0)+o(\|r-g_0\|^2)$; two sub-cases exhaust the possibilities.

\emph{(iii-a) $\nabla^2h(g_0)$ has a positive eigenvalue $\lambda_+$, unit eigenvector $v$.} Then $v^\top M_\gamma v=v^\top H_Sv+\gamma\lambda_+>0$ once $\gamma>\bar\gamma_{\mathrm{local}}\coloneqq|v^\top H_Sv|/\lambda_+$, finite and positive since $v^\top H_Sv<0$, by Assumption~2.2(ii), which makes $H_S(\theta_0)$ negative definite, and $g_0$ is then not even a local maximizer.

\emph{(iii-b) $\nabla^2h(g_0)\preceq0$, which includes the case $\nabla^2h(g_0)=0$ where the expansion is inconclusive at every order the hypotheses control.} Then $M_\gamma\prec0$ for every $\gamma>0$ and $g_0$ stays a local maximizer, so only a global comparison decides: by NT1 pick an interior $r_1$ with $\Delta h\coloneqq h(r_1)-h(g_0)>0$ and put $\Delta S\coloneqq\Sbar(g_0;\theta_0)-\Sbar(r_1;\theta_0)>0$ by strict properness; then $V(r_1;\theta_0,\gamma)-V(g_0;\theta_0,\gamma)=-\Delta S+\gamma\Delta h>0$ whenever $\gamma>\bar\gamma_{\mathrm{global}}(\theta_0)\coloneqq\Delta S/\Delta h$.

With $\bar\gamma(\theta_0)\coloneqq\min\{\bar\gamma_{\mathrm{local}},\bar\gamma_{\mathrm{global}}(\theta_0)\}$, and $\bar\gamma_{\mathrm{local}}=\infty$ in (iii-b), truthful reporting fails at $\theta_0$ for every $\gamma>\bar\gamma(\theta_0)$.
\end{proof}

\subsection{Theorem 4.2: optimal oversight is non-affine}

\paragraph*{Theorem 4.2 (Optimal Oversight Non-Affinity, restated under its weakest hypothesis).}
Assume (A1) and weak feasibility $\gamma/\beta \le D_G(p_{\min},1)$. Let the principal choose an approval rule $q : [0,1] \to [0,1]$ to maximize expected utility under Stackelberg timing, given that the agent best-responds, and suppose the type distribution $F$ places positive mass on both sides of the first-best threshold $p_{\min}$. Then no affine $q$ is optimal for the principal. Feasibility is used only to pass from screening to optimality: that no affine $q$ attains first-best \emph{screening} holds in every regime.

\begin{proof}
Assume (A1). Let $q(r)=a+br$ be affine, map $[0,1]$ into $[0,1]$, and attain first-best screening. The cap of Lemma~4.6 gives $q\leq\mathrm{ramp}$ on $[0,1]$, hence $q\equiv0$ on $[0,p_{\min}]$, which forces $b=0$ and so $q\equiv0$, whence $\tau\equiv0$, contradicting $\tau=1$ on the positive-measure set $p>p_{\min}$. So no affine rule attains first-best screening, and this half is regime-free. Assume in addition weak feasibility. Then the step rule of Theorem~4.5(i) attains first-best screening, hence $W^*$ by~(4.4), while the pointwise bound $\tau\Pi\leq\max(\Pi,0)$ inside~(4.2) gives $W\leq W^*$ for every rule; so an optimal rule exists, attains first-best screening by~(4.4) again with $A>0$ and $f>0$, and is not affine. In the saturated regime no rule attains first-best screening at all (Theorem~A.2(ii)) and the optimum is not characterized here.
\end{proof}

\subsection{Theorem 4.5: the reserve rule attains first best}

\paragraph*{Theorem 4.5 (The reserve rule attains first best, restated under its weakest hypothesis).}
Assume (A1) and weak feasibility $\gamma/\beta \le D_G(p_{\min},1)$, and let $r_0 \in (p_{\min},1]$ be the unique solution of $D_G(p_{\min},r_0) = \gamma/\beta$. Then:
\begin{enumerate}
\item[(i)] The step rule $q(r) = \mathbf 1\{r \ge r_0\}$ induces $\tau(p) = \mathbf 1\{p \ge p_{\min}\}$ for every $p \ne p_{\min}$, with strict optimality of the induced report; at $p = p_{\min}$ the agent is exactly indifferent between $r = p_{\min}$ and $r = r_0$, a single ($F$-null) type. Hence the step rule attains first-best welfare under any tie-breaking.
\item[(ii)] The reserve is an affine function of $\sqrt{\gamma/\beta}$, that is $r_0 = p_{\min} + c_0\sqrt{\gamma/\beta}$ for a constant $c_0$ independent of $(p_{\min},\gamma/\beta)$ across the whole feasible range, if and only if $G$ is quadratic, and $c_0 = 1$ if and only if $G'' \equiv 2$, the Brier normalization.
\end{enumerate}

\begin{proof}[Proof of part~(i)]
Assume (A1) and weak feasibility. \emph{Existence and uniqueness of $r_0$:} $C$ is continuous on $[p_{\min},1)$ with $C(p_{\min})=0$, strictly increasing, and $\lim_{r\uparrow1}C(r)=C(1)\geq\gamma/\beta$ in $[0,+\infty]$, so the intermediate value theorem on $[p_{\min},1)$ and strict monotonicity give a unique $r_0\in(p_{\min},1]$ with $C(r_0)=\gamma/\beta$. Writing the limit rather than an endpoint value admits the log score, for which $C(1)=+\infty$. No divergence of the calibration cost is used, and under bounded curvature none is available.

The step rule $q = \mathbf 1\{r \ge r_0\}$ satisfies the hypotheses of Lemma~4.4: it vanishes on $[0,r_0)$, where the ramp is non-negative, and equals $1 = \mathrm{ramp}$ on $[r_0,1]$. So $\tau = \mathbf 1\{p \ge p_{\min}\}$ off $p_{\min}$ with a unique best response, at $p_{\min}$ the argmax is $\{p_{\min},r_0\}$, and first-best welfare follows under any tie-breaking.
\end{proof}

\begin{proof}[Proof of part~(ii)]
Assume (A1). Suppose $r_0=p_{\min}+c_0\sqrt{\gamma/\beta}$ for a constant $c_0$ and all feasible $(p_{\min},\gamma/\beta)$. Fixing $p_{\min}$ and varying $\gamma/\beta$ over its whole feasible range gives $C(r)=(r-p_{\min})^2/c_0^2$ for every $r\in(p_{\min},1)$, that is $G'(r)(r-p_{\min})-[G(r)-G(p_{\min})]=(r-p_{\min})^2/c_0^2$; the left side is $(r-p_{\min})^2$ times the derivative of the difference quotient $r\mapsto[G(r)-G(p_{\min})]/(r-p_{\min})$, so that quotient equals $r/c_0^2+k$ for a constant $k$ and $G(r)=G(p_{\min})+(r-p_{\min})(r/c_0^2+k)$ is quadratic on $(p_{\min},1)$; letting $p_{\min}\downarrow0$ propagates it to $(0,1)$. No second derivative is used, so the step stays inside (A1). Conversely $G''\equiv2/c_0^2$ gives $D_G(p,r)=(p-r)^2/c_0^2$ and $r_0=p_{\min}+c_0\sqrt{\gamma/\beta}$. So the reserve is an affine function of $\sqrt{\gamma/\beta}$ if and only if $G$ is quadratic, and $c_0=1$ if and only if $G''\equiv2$, the Brier normalization; the scale must be named, since rescaling $G$ is observationally equivalent to rescaling $\beta$.
\end{proof}

\subsection{Proposition 4.15: exact screening at reduced intensity}

\paragraph*{Proposition 4.15 (Exact screening at reduced intensity, restated).}
Assume (A1) and $\gamma/\beta > D_G(p_{\min},1)$. Put $\lambda_{\mathrm{sat}} \coloneqq \beta D_G(p_{\min},1)/\gamma \in (0,1)$ and $q_{\mathrm{sat}} \coloneqq \lambda_{\mathrm{sat}}\,\mathbf 1\{r \ge 1\}$. Then $q_{\mathrm{sat}}$ induces $\tau(p) = \lambda_{\mathrm{sat}}\,\mathbf 1\{p \ge p_{\min}\}$ for every $p \ne p_{\min}$, so the screening boundary lands exactly at the first-best cutoff, and
\[
W(q_{\mathrm{sat}}) \;=\; u_d + \lambda_{\mathrm{sat}}\int_{p_{\min}}^1 \Pi(p) f(p)\, dp \;>\; u_d .
\]
This is optimal within the class of approval rules that reject almost every type below $p_{\min}$.

\begin{proof}
Let $\gamma/\beta>D_G(p_{\min},1)$, so no reserve exists, and put $\lambda_{\mathrm{sat}}\coloneqq\beta D_G(p_{\min},1)/\gamma\in(0,1)$ and $q_{\mathrm{sat}}=\lambda_{\mathrm{sat}}\mathbf 1\{r\geq1\}$. Here $\bar r=1$ and $\bar q=\lambda_{\mathrm{sat}}$, and $q_{\mathrm{sat}}$ satisfies (C1) and (C2) of the extended Lemma~4.4: it vanishes below $r=1$, where the ramp is non-negative, and equals $\mathrm{ramp}(1)=\lambda_{\mathrm{sat}}$ at $r=1$. So $\tau=\lambda_{\mathrm{sat}}\mathbf 1\{p\geq p_{\min}\}$ and $W(q_{\mathrm{sat}})=u_d+\lambda_{\mathrm{sat}}\int_{p_{\min}}^1\Pi f>u_d$. For optimality within the rules that reject almost every type below $p_{\min}$, both cap inequalities of Lemma~4.6 are available in the weaker form noted in its proof, which this class supplies and which does not presuppose that first best is attainable. Split on the best response: if $r^*(p)<p_{\min}$ then $q(r^*(p))=0$, so $\tau(p)=0$; if $r^*(p)\geq p_{\min}$ then $\gamma\tau(p)\leq\beta C(r^*(p))\leq\beta C(1)$ because $C$ is increasing there. Either way $\tau\leq\lambda_{\mathrm{sat}}$, whence $W(q)\leq W(q_{\mathrm{sat}})$.
\end{proof}

\subsection{Theorem A.1: canonical form of the direct mechanisms}

\paragraph*{Theorem A.1 (Canonical form of every incentive-compatible direct mechanism, restated).}
$(x,t)$ is incentive compatible if and only if there is a convex $\Psi:[0,1]\to\R$, with $\Psi'$ a selection of its subdifferential, such that
\begin{equation*}
t(r,\omega)=\underbrace{\Psi(r)+\Psi'(r)(\omega-r)}_{\text{proper score with generator }\Psi}-\underbrace{\gamma\,x(r)}_{\text{approval charge}} .
\end{equation*}
The rule $x$ is unrestricted, $\Psi$ is unrestricted among convex functions admitting a finite subgradient at every point of $[0,1]$, the two are chosen independently, and the agent's on-path payoff is $\Psi(p)$.

\begin{proof}
\emph{Sketch.} Write $B(r)\coloneqq t(r,1)-t(r,0)$ and $\alpha(r)\coloneqq t(r,0)+\gamma x(r)$, so $u(r,p)=\alpha(r)+pB(r)$ is affine in the type. The argument is the standard envelope characterization of incentive compatibility in a one-dimensional linear environment \citep{myerson1981optimal,rochet1987necessary,vohra2011mechanism}, given in full because the form used below holds at every pair of types rather than almost everywhere; the theorem adds to it the reading of~(A.2), the first bracket as a proper score with generator $\Psi$ and the second as an approval charge of exactly $\gamma x(r)$.

\emph{Step 1 (necessity: $B$ is monotone at every pair).} Let $U(p)\coloneqq\max_ru(r,p)$, convex as a pointwise maximum of affine functions of $p$. If the mechanism is incentive compatible then $U(p)=\alpha(p)+pB(p)$; adding the incentive constraints at $p$ against $p'$ and at $p'$ against $p$ gives $(p'-p)(B(p')-B(p))\geq0$, so $B$ is non-decreasing at every pair, not only where $U$ is differentiable, and $B(r)\in\partial U(r)$ throughout.

\emph{Step 2 (necessity: the envelope integral).} $U$ is a pointwise maximum of affine functions, hence convex and lower semicontinuous, so continuous on $[0,1]$, and $U'=B$ off a countable set gives $U(p)=U(0)+\int_0^pB$, and $\alpha(r)=U(r)-rB(r)$ for every $r$.

\emph{Step 3 (sufficiency).} If $B$ is non-decreasing and $\alpha(r)=U(r)-rB(r)$ with $U(r)\coloneqq U(0)+\int_0^rB$, then $U$ is convex with $U'=B$, so $u(r,p)=U(r)+B(r)(p-r)\leq U(p)$ by the supporting-hyperplane inequality, with equality at $r=p$.

\emph{Step 4 (the canonical form).} Setting $\Psi\coloneqq U$ gives $t(r,1)-t(r,0)=\Psi'(r)$ and $t(r,0)=\Psi(r)-r\Psi'(r)-\gamma x(r)$, which is~(A.2) at $\omega\in\{0,1\}$.
\end{proof}

\subsection{Theorem A.2: welfare equivalence and truthfulness inequivalence}

\paragraph*{Theorem A.2 (Welfare equivalence, truthfulness inequivalence, restated).}
Fix the environment, a strictly convex $G\in C^2$, and $\beta,\gamma>0$. \emph{(i)} The unrestricted benchmark attains $W^*$ with a strictly incentive-compatible mechanism, and $W^*$ bounds every mechanism in either class. \emph{(ii)} The restricted architecture attains $W^*$ if and only if $\gamma/\beta\leq D_G(p_{\min},1)$, and when it does, both the step rule of Theorem~4.5 and the Bregman ramp of Theorem~4.7 attain it. \emph{(iii)} No member of the restricted class attaining $W^*$ is incentive compatible; the incentive-compatible members are exactly the constant rules, all strictly below $W^*$.

\begin{proof}
\emph{Step 1 (part (i)).} This is the first coordinate reading that follows Theorem~A.1 in Appendix~A, with the pointwise bound in~(4.2).

\emph{Step 2 (part (ii)).} Sufficiency is Theorem~4.5. For necessity, the cap of Lemma~4.6 gives $\gamma q\leq\beta D_G(p_{\min},\cdot)$ pointwise, and any $p>p_{\min}$ with $\tau(p)=1$ supplies an $r$ with $q(r)=1$; the same cap puts $q\equiv0$ on $[0,p_{\min}]$, so $r>p_{\min}$, and $r\mapsto D_G(p_{\min},r)$ is increasing above $p_{\min}$, whence $\gamma\leq\beta D_G(p_{\min},r)\leq\beta D_G(p_{\min},1)$.

\emph{Step 3 (part (iii): truthfulness forces a constant rule).} We first show that $q$ induces truthful reporting for every type in $[0,1]$ if and only if $q$ is constant. Sufficiency is immediate, every deviation costing $\beta D_G(p,r)>0$ and buying nothing. For necessity, fix a compact $[a,b]\subset(0,1)$ on which $G''\leq\bar G''_{[a,b]}$; if $r=p$ maximizes the agent's payoff for both $p$ and $r$ in $[a,b]$, comparing from each side gives $\gamma[q(r)-q(p)]\leq\beta D_G(p,r)$ and $\gamma[q(p)-q(r)]\leq\beta D_G(r,p)$, while $D_G(p,r)=\int_p^rG''(z)(z-p)dz\leq\tfrac12\bar G''_{[a,b]}(r-p)^2$ and symmetrically, so $|q(r)-q(p)|\leq\frac{\beta\bar G''_{[a,b]}}{2\gamma}(r-p)^2$ on $[a,b]$, a quadratic H\"older condition forcing $q$ constant on $(a,b)$; exhausting $(0,1)$ by such intervals gives $q$ constant at some $c$ there; truthfulness at types approaching an endpoint, where $D_G(p,\cdot)$ vanishes, forces $q(0)\leq c$ and $q(1)\leq c$, and truthfulness of the endpoint types themselves, which can deviate to interior reports at vanishing cost, forces $q(0)\geq c$ and $q(1)\geq c$.

\emph{Step 4 (part (iii): a constant rule is strictly below first best).} Such a $q$ induces $\tau\equiv c$, with $W(q)=u_d+c\int_0^1\Pi f$, whose best value over $c\in[0,1]$ is strictly below $W^*=u_d+\int_{p_{\min}}^1\Pi f$ because $\int_0^{p_{\min}}\Pi f<0$; and a constant $\tau$ cannot equal $\mathbf 1\{p\geq p_{\min}\}$ on a set of positive measure on both sides of $p_{\min}$.
\end{proof}

\subsection{Corollary 4.10: the first-best rules}

\paragraph*{Corollary 4.10 (The first-best rules, restated under its weakest hypothesis).}
Assume (A1) and weak feasibility, and let $q : [0,1] \to [0,1]$ be non-decreasing and admit a global best response at every type (automatic if $q$ is upper semicontinuous). Then $q$ attains first-best screening if and only if $q \le \mathrm{ramp}$ on $[0,1]$ and $q(r_0) = 1$. The step rule of Theorem~4.5 and the Bregman ramp of Theorem~4.7 are two named members of that family; under (S) the truncated line $q_\Lambda$ of Theorem~4.8(ii), for any $\Lambda \ge \Lambda_c(G)$, is a third.

\begin{proof}[Proof of Corollary~4.10]
\emph{Necessity.} The cap is Lemma~4.6, which gives $q\leq\mathrm{ramp}<1$ on $[0,r_0)$. Pick $p\in(p_{\min},r_0)$ with $\tau(p)=1$, possible since that interval carries positive $F$-measure; its best response, which exists by convention, lands at some $r^*(p)$ with $q(r^*(p))=1$ and hence $r^*(p)\geq r_0$. Under weak-but-not-strict feasibility $r_0=1$ and $q(r_0)=1$ follows at once. Otherwise the saturation statement of Lemma~4.6 supplies $r_n\downarrow r_0$ with $q(r_n)=1$, so $V(r^*(p);p)\geq\gamma-\beta D_G(p,r_n)$ for every $n$, whence $D_G(p,r^*(p))\leq D_G(p,r_0)$ by continuity of $D_G(p,\cdot)$; since $D_G(p,\cdot)$ is strictly increasing above $p$ and $r^*(p)\geq r_0>p$, this forces $r^*(p)=r_0$ and $q(r_0)=1$.

\emph{Sufficiency.} A non-decreasing $q$ with $q(r_0)=1$ and $q\leq1$ satisfies $q\equiv1=\mathrm{ramp}$ on $[r_0,1]$, so (C1) and (C2) of the extended Lemma~4.4 hold and the lemma applies.
\end{proof}

\subsection{The closed form for the critical slope}

\paragraph*{Remark (Closed form under monotone ramp slope).}
Assume (S) and write $s(r_0) = (\beta/\gamma) G''(r_0)(r_0 - p_{\min})$ for the ramp's slope at the reserve. Then $\Lambda_c(G) = s(r_0)$ if and only if $C$ lies above its tangent at $r_0$ throughout $[p_{\min},r_0]$, that is $C(a) \ge C(r_0) - C'(r_0)(r_0-a)$ for all $a \in [p_{\min},r_0]$, and $\Lambda_c(G) > s(r_0)$ strictly otherwise. A checkable \emph{sufficient} condition is that $C$ be convex on $[p_{\min},r_0]$, equivalently that the ramp slope $s(r) = (\beta/\gamma)G''(r)(r-p_{\min})$ be non-decreasing there, equivalently $G'''(r)(r-p_{\min}) + G''(r) \ge 0$ under (A3). In all cases $\Lambda_c \ge 1/m$.

\begin{proof}
The inequality $\Lambda_c\geq s(r_0)$ always holds, since the chord slope tends to $(\beta/\gamma)C'(r_0)=s(r_0)$ as $a\uparrow r_0$, as established in the proof of Theorem~4.8(i). So $\Lambda_c=s(r_0)$ if and only if every chord slope to $(r_0,1)$ is at most $s(r_0)$, i.e. $1-\mathrm{ramp}(a)\leq s(r_0)(r_0-a)$ for all $a$, i.e., multiplying by $\gamma/\beta$, $C(r_0)-C(a)\leq C'(r_0)(r_0-a)$: the requirement that $C$ lie above its tangent at $r_0$ on $[p_{\min},r_0]$. Convexity of $C$ there, equivalently monotonicity of $s$, equivalently $G'''(r)(r-p_{\min})+G''(r)\geq0$ under (A3), implies tangent support, a convex function lying above each of its tangents; the converse fails, the condition constraining $C$ only relative to the single tangent at $r_0$. The chord at $a=p_{\min}$ has $\mathrm{ramp}(p_{\min})=0$ and length $m$, so $\Lambda_c\geq1/m$ in all cases.
\end{proof}

The sufficient condition holds for every scoring rule in practical use: for $G = p^k$ with $k > 1$, $s(r) \propto r^{k-2}(r-p_{\min})$ has derivative proportional to $r^{k-3}[(k-1)r - (k-2)p_{\min}] > 0$ on $(p_{\min},1)$; for Brier $s$ is affine and increasing; and the log score is treated in Example~4.14 of the main paper. The closed form is thus available wherever a practitioner needs it, while Theorem~4.8 rests on the chord bound and never on the closed form.

The distinction is not vacuous. For the generator $G''(r) = 2 + 100\exp(-((r-0.45)/0.02)^2)$ at $p_{\min} = 0.4$, $\gamma/\beta = 0.2$, whose curvature spikes strictly inside $(p_{\min},r_0)$ and falls away again before the reserve, $\Lambda_c$ exceeds $s(r_0)$ by a factor of $5.30$, and a truncated line of slope $1.02\,s(r_0)$ falsely approves every type below $p_{\min}$ (Table~1 of the main paper, row 3), while any line of slope at least $\Lambda_c$ screens exactly by Theorem~4.8(ii). The construction of Theorem~4.8(ii) therefore needs $\Lambda \ge \Lambda_c \ge 1/m$. No standard score has such a spike.

\subsection{Hypothesis traceability}

The main paper states every result of its Section~4 under the standing hypothesis (S), that is (A2) and strict feasibility, except Proposition~4.15, which lies outside it, and says after each result what weaker grade suffices. The table below collects those statements: for each result, the hypothesis under which the proof above or in Appendix~A of the main paper is written, and where a part needs more than the rest.

\begin{table}[h]
\caption{The weakest hypothesis each proof uses. (A1): $G$ continuous, strictly convex, $C^1$ on $(0,1)$ with the endpoint extension of $D_G$; (A2): (A1) and $G \in C^2((0,1))$, $G'' > 0$; (A3): (A2) and $G \in C^3((0,1))$; (S): (A2) and strict feasibility $\gamma/\beta < D_G(p_{\min},1)$; weak feasibility: $\gamma/\beta \le D_G(p_{\min},1)$.}
\centering
\small
\begin{tabular}{@{}p{4.3cm}p{3.6cm}p{5.0cm}@{}}
\toprule
Result & Proof uses & Parts needing more \\
\midrule
Theorem 4.2 (non-affinity) & (A1) & weak feasibility for optimality; the screening half is regime-free \\
Lemma 4.4 (envelope) & (A1) & none; the extended form covers the saturated regime \\
Theorem 4.5 (reserve rule) & (A1), weak feasibility & none \\
Lemma 4.6 (structure) & (A1), weak feasibility & the cap needs no feasibility and only $\tau = 0$ a.e.\ below $p_{\min}$ \\
Theorem 4.7 (ramp) & (A1), weak feasibility for part (i) & (S) for the Lipschitz constant and for part (ii) \\
Theorem 4.8 (critical slope) & (S) & none \\
Corollary 4.9 ($C^1$ impossibility) & (S) & strict feasibility cannot be weakened \\
Corollary 4.10 (first-best rules) & (A1), weak feasibility & (S) for the truncated-line clause \\
Proposition 4.11 (below $\Lambda_c$) & (A3), strict feasibility, $s$ non-decreasing, $s'(r_0) > 0$ & none \\
Proposition 4.12 (distribution-free) & as the results it cites & the reserve under weak feasibility alone \\
Proposition 4.13 and A.3 (score selection) & the hypotheses of Proposition 4.11 & none \\
Proposition 4.15 (saturated regime) & (A1), $\gamma/\beta > D_G(p_{\min},1)$ & none \\
Theorem A.1 (canonical form) & the environment only & none \\
Theorem A.2 (welfare equivalence) & $G \in C^2$ strictly convex & none \\
Closed-form remark above & (S) & (A3) for the derivative criterion \\
\bottomrule
\end{tabular}
\end{table}
\section{Multi-Dimensional Types}\label{sec:multi-dim}

The impossibility half of the paper is untouched by $d$-dimensional types, being proved in $\R^d$ to begin with. The design half loses the generator-independence it has on the line, because an agent holding a $d$-dimensional report chooses a direction of inflation as well as a distance and the generator's Hessian field $H_G \coloneqq \nabla^2 G$ prices directions unequally. A hard rule becomes $q(r) = \mathbf{1}\{r \in \mathcal{A}\}$ for an acceptance region $\mathcal{A} \subseteq \calR$, so the reserve is a level set of the divergence rather than a point, and the principal must place it so that the types willing to inflate into $\mathcal{A}$ are exactly the first-best set $\mathcal{A}^* \coloneqq \{p \in \Theta : W_A(p) \ge u_d\}$, with $\Theta = \calR$. Under the Brier score $H_G$ is isotropic and the scalar answer carries exactly: $\mathcal{A}$ is the $m$-inner contraction of $\mathcal{A}^*$ with $m \coloneqq \sqrt{\gamma/\beta}$, exact when every boundary point of $\mathcal{A}^*$ is touched by an $m$-ball inside it, which linear welfare supplies away from the boundary of the type space and convexity alone does not. Away from Brier the agent's cost is an asymmetric Bregman divergence rather than a distance, so the shape the principal contracts by varies across the type space. \multidimnote{}

The main open question is the multi-dimensional analogue of the critical slope. In the scalar theory steepness rather than regularity is the binding constraint on a continuous approval rule (Theorem~\ref{thm:critical-slope} and Proposition~\ref{prop:below-threshold}), and $\Lambda_c(G)$ is a supremum of chord slopes of the Bregman ramp, each chord on the line having an unambiguous horizontal extent. In $\R^d$ the constraint governs how fast $q$ may climb along a path in report space, saturation occurs on a level set of $D_G(p_{\min},\cdot)$ rather than at a point, and the natural candidate is an infimum over paths of a direction-weighted climb rate, a Finsler-type quantity built from the Hessian field rather than a number read off one chord. Whether that infimum has a closed form, whether it collapses to the ramp's slope at the reserve as the scalar constant does for the standard generators, and whether the cubic law survives in $\R^d$ are open; the vector-valued elicitability framework of \citet{fissler2016higher} is the natural ambient theory for them.

\section{Multi-Dimensional Types: The Geometry of the Acceptance Region}\label{sec:multi-dim-geom}

In the scalar setting the oversight design problem is generator-independent in a strong sense: for every admissible generator a hard threshold at the reserve $r_0(G)$ screens types exactly, and so does a ramp of slope at least $\Lambda_c(G)$ (Theorems~4.5, 4.7 and~4.8). Multi-dimensional types leave the impossibility half of the paper untouched, since it is proved in $\R^d$ to begin with, but they break that independence for a geometric reason: an agent holding a $d$-dimensional report chooses a direction of inflation as well as a distance, and the score prices directions unequally.

The escape mechanism itself is dimension-free. A hard rule becomes $q(r) = \mathbf{1}\{r \in \mathcal{A}\}$ for an acceptance region $\mathcal{A} \subseteq \calR$, the agent's decision is still binary, and that binary choice still partitions the type space into inflating and non-inflating types for any $d$ and any $G$. The object the principal must choose changes. In one dimension the reserve is a single number obtained by solving $D_G(p_{\min},r_0) = \gamma/\beta$, and any target cutoff is reachable because the calibration cost is monotone in the report; in $\R^d$ the principal chooses an entire surface $\partial \mathcal{A}$, and must choose it so that the types willing to inflate into $\mathcal{A}$ are exactly the first-best set $\mathcal{A}^* \coloneqq \{p \in \Theta : W_A(p) \ge u_d\}$, with $\Theta = \calR$.

Under the Brier score the problem has a clean solution, because the score is isotropic. With $G(p) = \|p\|^2$ the Hessian field $H_G \coloneqq \nabla^2 G$ is $2I_d$ everywhere, so the cost of moving a report from $p$ into $\mathcal{A}$ is $\beta \cdot d(p,\mathcal{A})^2$ with $d(p,\mathcal{A}) = \inf_{a \in \mathcal{A}}\|a-p\|$, the agent inflates to the nearest point of $\mathcal{A}$, and does so exactly when $d(p,\mathcal{A}) \leq m \coloneqq \sqrt{\gamma/\beta}$. The approved set is the $m$-neighborhood $\mathcal{A} \oplus m B$, so the principal takes $\mathcal{A}$ to be the $m$-inner contraction of $\mathcal{A}^*$, which recovers the scalar reserve offset with a surface in place of a point.

Exactness then turns on a condition invisible in one dimension: contracting $\mathcal{A}^*$ by $m$ and re-expanding must return $\mathcal{A}^*$, which holds precisely when every boundary point of $\mathcal{A}^*$ is touched by an $m$-ball contained in $\mathcal{A}^*$. Linear welfare supplies it away from the boundary of the type space, since an affine $W_A$ makes $\mathcal{A}^*$ the part of a half-space inside $\Theta$, whose inner contraction is a parallel translate of the separating hyperplane except where that hyperplane meets $\partial\Theta$. Convexity alone does not: re-expanding the inner contraction of a convex polytope rounds its corners, so types just inside a corner of $\mathcal{A}^*$ are rejected.

Away from the Brier score the cost is not a distance at all. The agent's calibration loss is the Bregman divergence $D_G(p,r)$, asymmetric in its arguments, whose sublevel sets $\{r : D_G(p,r) \leq \gamma/\beta\}$ are shaped by the local eigenstructure of $H_G$ and change from type to type. The principal is then fitting $\mathcal{A}^*$ with a structuring shape that varies across the type space rather than offsetting it by a fixed ball, and both the direction in which the agent inflates and the price it pays for that direction are set by $H_G$. None of this appears in the scalar case: direction is not a choice, only the distance to the reserve matters, and the attainment results hold for every generator with no curvature hypothesis at all.
\end{appendix}

\paragraph*{Artifact availability.}
Every computed numeral in this paper carries its reproduction parameters in Table~\ref{tab:numerics}, and each is recomputed in a public archive from the defining objects rather than from the derivations it checks. The reserves, the critical slopes and the ramp slopes at the reserve, the chord anchor and the ratio of the two slopes, the wrong-threshold welfare cost and the slope premiums are enclosed in Arb ball arithmetic \citep{johansson2017arb}, in which a comparison holds only when it holds for every point of every ball, so those enclosures are proofs rather than agreements between floating-point runs; the two grid-computed illustrations (the mollified kink after Theorem~\ref{thm:kink} and the screening errors of Table~\ref{tab:numerics}, row 3), the constants of Proposition~\ref{prop:below-threshold} and the instance and property checks of results proved here are floating-point. The archive is at \url{https://github.com/Future-Computing-Group/endomiscal-certificates}, deposited at \href{https://doi.org/10.5281/zenodo.22983104}{DOI~10.5281/zenodo.22983104}.

\begin{acks}
The authors thank colleagues at the Future Computing Group, University of
Oulu, and the Department of Computer Science, University of Helsinki, for
feedback on the marketplace and AI agent oversight instances. Manuscript preparation used Claude
(Anthropic) for drafting assistance.

This work was supported by the Research Council of Finland through the 6G
Flagship programme (grant 318927), the Strategic Research Council affiliated
with the Academy of Finland through the CO2CREATION project (grant 372355),
by Business Finland through the Neural pub/sub research project (diary number
8754/31/2022), and by the European Regional Development Fund (ERDF; project
numbers A81568, A91867).
\end{acks}

\bibliographystyle{ACM-Reference-Format}
\bibliography{references}

\end{document}